\documentclass[a4paper,final]{siamltex}
\usepackage{amsmath,amsfonts,amssymb}
\DeclareMathAlphabet\scr{U}{scr}{m}{n}
\SetMathAlphabet\scr{bold}{U}{scr}{b}{n}
\DeclareFontFamily{U}{scr}{\skewchar\font'177}
\DeclareFontShape{U}{scr}{m}{n}{<-6>rsfs5<6-8>rsfs7<8->rsfs10}{}
\DeclareFontShape{U}{scr}{b}{n}{<-6>rsfs5<6-8>rsfs7<8->rsfs10}{}

\def\softl{l\kern-0.35ex\raise0.1ex\hbox{'}\kern-0.15ex}

\newtheorem{assumption}[theorem]{Assumption}
\newtheorem{remark}[theorem]{Remark}
\newtheorem{example}[theorem]{Example}
\newcommand{\simple}{\ensuremath{\mathcal H}}
\newcommand{\bounded}{\ensuremath{\mathcal{H}^{b}}}
\newcommand{\boundedU}{\ensuremath{\mathcal{H}^{\hat{U}}}}
\newcommand{\boundedW}{\ensuremath{\mathcal{H}^{W}}}
\newcommand{\adm}{\ensuremath{\overline{\mathcal H}}}
\newcommand{\super}{\ensuremath{\mathcal H}^{\mathrm{s}}}
\newcommand{\measures}{\ensuremath{\mathcal{M} \cap P_V}}
\newcommand{\emeasures}{\ensuremath{\mathcal{M}^e \cap P_V}}
\newcommand{\optf}{\ensuremath{ \widehat{f}}}

\newcommand{\optH}{\ensuremath{ \widehat{H}}}
\newcommand{\optQ}{\ensuremath{ \widehat{Q}}}
\newcommand{\opty}{\ensuremath{ \widehat{y}}}
\newcommand{\Uhat}{{\widehat{U}}}
\newcommand{\SM}{\scr{S}^{M^{\Psi}}}
\newcommand{\xbar}{\ensuremath{\overline{x}}}
\newcommand{\I}{\mathcal{I}}
\newcommand{\SsigPsi}{\scr{S}\!\!_{\sigma}\!^\Psi}
\newcommand{\SIPsi}{\scr{S}_{\!\I}^\Psi}
\newcommand{\SsigMPsi}{\scr{S}\!\!_{\sigma}\!^{M^{\Psi}}}
\newcommand{\SIMPsi}{\scr{S}_{\!\I}^{M^\Psi}}
\newcommand{\SsigU}{\scr{S}\!\!_{\sigma}\!^{\Uhat}}
\newcommand{\SsigMU}{\scr{S}\!\!_{\sigma}\!^{M^{\Uhat}}}
\newcommand{\SIU}{\scr{S}_{\!\I}^{\Uhat}}
\newcommand{\SIMU}{\scr{S}_{\!\I}^{M^{\Uhat}}}
\newcommand{\Orlicz}{\ensuremath{ L^{\Uhat}}}

\newcommand{\nada}[1]{}

\begin{document}
\title{Admissible strategies
in  semimartingale portfolio selection\thanks{To appear in SIAM J. Control Optim.}}
\author{ Sara Biagini\thanks{University of Pisa ({\tt sara.biagini@ec.unipi.it}). Part of this research was conducted 
while visiting Collegio Carlo Alberto in Moncalieri, Turin, Italy in Spring 2009. Warm hospitality and financial support 
of the Collegio are gratefully acknowledged.}
 \and Ale\v{s} \v{C}ern\'{y}\thanks{Cass Business School, City University London ({\tt Ales.Cerny.1@city.ac.uk}).}} 
\maketitle
\begin{center}
\emph{\small Dedicated to Walter Schachermayer on the occasion of his 60th birthday.}\vspace*{0.5cm}\\
\end{center}
\begin{abstract}
The choice of admissible trading strategies in mathematical
modelling of financial markets is a delicate issue, going back to
Harrison and Kreps \cite{hk79}. In the context of optimal
portfolio selection with expected utility preferences this question has been the focus of considerable
attention over the last twenty years.

We propose a novel notion of admissibility that has many pleasant
features -- admissibility is characterized purely under the
objective measure $P$;  each admissible strategy can be approximated
by simple strategies using finite number of trading dates;
the wealth of any admissible strategy is a
supermartingale under all pricing measures; local boundedness of
the price process is not required; neither strict monotonicity,
strict concavity nor differentiability of the utility function
are necessary;  the definition encompasses both the classical
mean-variance preferences and the monotone expected utility.

For utility functions finite on $\mathbb{R}$, our class represents a minimal set containing
simple strategies which also contains the optimizer, under
conditions that are milder than the celebrated
\emph{reasonable asymptotic elasticity} condition on the utility
function.
\end{abstract}


\begin{AMS}primary 60G48, 60G44, 49N15, 91B16; secondary 46E30, 46N30\end{AMS}

{\footnotesize\textbf{\!\!JEL subject classifications.} G11, G12, G13}

\pagestyle{myheadings}
\markboth{S. BIAGINI AND A. \v{C}ERN\'{Y}}{ADMISSIBILITY IN SEMIMARTINGALE PORTFOLIO SELECTION} 
 
\section{Introduction}

A central concept of financial theory is the notion of
a self-financing investment strategy $H$,  whose discounted wealth  is expressed
mathematically by the stochastic integral
$$ x+H\cdot S_t:=x+\int_{(0,t]} H_sdS_s,$$ where $S$ is a semimartingale process on a stochastic basis $(\Omega, (\mathcal{F}_t)_{0\leq t \leq T}, P)$,
representing discounted prices of $d$ traded assets, and $x$ is the
initial wealth.

Stochastic integration theory formulates minimal requirements for the integral above to exist,
see Protter \cite{Pr05}. The class of predictable processes $H$ for
which the integral exists is denoted by $L(S;P)$ or simply $L(S)$.
However,  the \emph{whole} of  $L(S)$ is
\emph{not appropriate} for financial applications. Specifically,
Harrison and Kreps \cite{hk79} noted that when  all processes in $L(S)$ are allowed as trading strategies, arbitrage
opportunities arise even in the standard Black-Scholes model. This
is not a problem of the model $S$ -- the reason is that the theory of
stochastic integration   operates with a set of integrands
far too rich for such applications.   The solution proposed by the subsequent no-arbitrage literature, see \cite{Sch94,ds98},
is to restrict attention to a subset $\bounded\subseteq L(S)$ of strategies whose wealth is bounded uniformly from below by a constant.

Now consider a concave non decreasing utility function $U$ and an agent who wishes to maximize the expected utility of her terminal wealth, $E[U(x+H\cdot S_T)]$.
In this context,  $\mathcal{A}\subseteq L(S)$ will be a good set
of trading strategies if the utility maximization over $H\in\mathcal{A}$ is well posed and if $\mathcal{A}$ contains the
optimizer,
$$U(+\infty)>\sup_{H\in\mathcal{A}} E[U(x+H\cdot S_T)] = \max_{H\in\mathcal{A}} E[U(x+H\cdot S_T)].$$

Historically, the search for a good definition of admissibility has proved to be a difficult task
and it has evolved in two streams.
For utility functions finite on a half-line, for example a logarithmic utility, there is a natural
 definition:  admissible strategies are again those
 in $\bounded$, see \cite{ks99,CSW01,ks03}. Remarkably, this theoretical framework is valid for \emph{any} arbitrage-free $S$.

 For utility functions finite on the whole $\mathbb{R}$, the situation is more complicated.
 The definition of admissibility via $\bounded$ works only to a certain
 extent.  Here $S$ has to be locally bounded (or $\sigma$-bounded) to ensure that $\bounded$ is sufficiently rich for a duality framework to work,  cf. \cite{Sch01}.
 Moreover, the class $\bounded$ will typically fail to contain the optimizer -- this happens, for example,
in the classical Black-Scholes model under exponential utility.

 A possible choice in this situation is  to consider all strategies whose wealth is a martingale under all
 suitably defined pricing measures (see \S \ref{sect: simple}).
 This approach works well for exponential utility, cf. \cite{dgrsss02,kabstr02}. However, the
 seminal work of Schachermayer \cite{Sch03} shows that, for general utilities, the martingale class is too narrow to catch the optimizer. The optimal strategy only exists
 among strategies whose wealth is a \emph{supermartingale} under all pricing
 measures.  For this reason,   the \emph{supermartingale class} is  now considered
the  best  notion of admissibility.

It is evident from our discussion that admissibility is currently
defined in a primal way for utility functions
 finite on $\mathbb{R}_+$ but for utilities finite
 on $\mathbb{R}$ \emph{the definition is dual, via pricing measures}.
 A connection between the two approaches is foreshadowed in Schachermayer \cite{Sch01}
 who defines  a set of \emph{admissible terminal wealths}  as those positions  whose utility can be approximated in $L^1(P)$ by strategies with wealth bounded from below. Under suitable technical assumptions, the optimal wealth exists and there is a trading strategy in the supermartingale class which leads to the optimal wealth, see also Owen \cite{O02} and Bouchard et al. \cite{btz04}.

All of the papers above dealing with utility finite on $\mathbb{R}$ use \emph{locally bounded price
processes}.  Biagini and Frittelli \cite{bf05} employ a wider class of well-behaved price processes compatible with the utility $U$. In \cite{bf07} they show that for this class of price processes there is always an optimizer in Schachermayer's set of supermartingale strategies.  In a subsequent paper \cite{bf08},
they propose a unified treatment for utility functions finite on a
half-line as well as those finite on the whole $\mathbb{R}$, for an even wider class of semimartingales $S$.
As we show in \S \ref{sect: simple} their hypotheses on $S$  amount to our Assumption \ref{maximalS}.
In contrast to the present paper, \cite{bf08} use admissible strategies $\boundedW$
whose wealth is controlled from below by (a multiple of) an
exogenously given, fixed random variable $W>0$. When $W$ is constant, one recovers the usual set $\bounded$ of
strategies with wealth bounded uniformly from below. Here, too, the optimal strategy
may fail to be in $\boundedW$, there is no approximation result for the optimizer,
and when $S$ is not particularly well behaved the optimizer may in principle depend on the choice of the loss control $W$.

The philosophy of the present paper is to make the definition of admissibility general enough to provide
a ``unified  treatment'' of utility functions
in the spirit of \cite{bf08}, while keeping the definition as natural and intuitive as possible
by \emph{not} resorting to duality.  We use a bottom-up approach whereby we first define a class of well-behaved simple trading strategies $\simple$ which can be interpreted as buy-and-hold strategies over finitely many dates (see Definition \ref{simple strategies} for details). In the locally bounded case $\simple$ corresponds to buy-and-hold strategies whose wealth is uniformly bounded in absolute value. We then define admissible strategies $\adm$ as suitable limits of strategies in $\simple$.

\begin{definition}\label{adm}
  $H\in L(S)$ is an \emph{admissible integrand} if  $ U (H\cdot S_T) \in L^1(P)$ and if
  there exists an approximating sequence $(H^n)_n$ in
 $\simple $ such that:
\begin{romannum}
    \item   $H^n\cdot S_t \rightarrow  H\cdot S_t $ in probability
    for all $t \in [0,T]$;
    \item  $ U( H^n\cdot S_T )  \rightarrow  U( H\cdot S_T ) $ in
    $L^1(P)$.
\end{romannum}
The set of all admissible integrands is denoted  by $\adm$.
\end{definition}

 The two requirements above are natural assumptions
if considered \emph{separately}.
Item (i) is in the spirit of  the construction of the stochastic integral itself,
 while item (ii) ensures that utility of an admissible strategy can be approximated   by
 the utility from simple strategies.
  Definition \ref{adm} combines these  \emph{two desirable approximation features} together.

The key point of the present paper is that
we do not ask for approximation of terminal utility \emph{only}, as is done in \cite{Sch01, O02, btz04}, but
we also require an approximation of the wealth process  at intermediate times,
 as in \v{C}ern\'{y} and Kallsen \cite[Definition 2.2]{ck07}. What is more, our definition does not rely on regularity properties of $U$, such as strict concavity, strict monotonicity or differentiability.

Our results then follow rather smoothly:  $\adm$  is a subset of the supermartingale class (Proposition \ref{admsuper}) and the optimizer belongs to  $\adm$ under very mild conditions, as shown in the main  Theorem \ref{main theorem}. Therefore, as a byproduct, we also obtain an
extremely compact proof of the supermartingale property of the
optimal solution.

The paper is organized as follows. In \S 2.1-\S 2.3 there are
basic definitions from convex analysis, theory of Orlicz spaces
and stochastic integration. Section 2.4 contains a new result on
$\sigma$-lo\-ca\-li\-za\-tion. In \S 3.1 and \S 3.2 we discuss
conditions imposed on the price process $S$ and the corresponding
definitions of simple strategies. In \S 3.3 we prove the
martingale property of simple strategies. In \S 3.4 we define
the admissible strategies and prove their supermartingale
property.
 In \S 4.1 and \S 4.2 we  discuss the customary conditions
 of \emph{reasonable asymptotic elasticity}
and other related conditions used in the literature and we
contrast them with a weaker Inada
condition at $+\infty$ employed in this paper. The main result
(Theorem \ref{main theorem}) is stated and proved in \S 4.3.  Section 5 provides more details on the main assumptions and
on the advantages of our framework compared to the existing literature. Section 6 contains technical lemmata.

\section{Mathematical preliminaries}
\subsection{Utility functions}\label{SecU}

A utility function $U$ is a proper, concave, non-decreasing, upper
semicontinuous function. Its effective domain is the non-empty set
\begin{equation}\label{effective domain}
\mathrm{dom}\,U:=\{ x\mid U(x)>-\infty\}.
\end{equation}
The infimum  of the effective domain of $U$ is denoted by
\begin{equation}\label{x_under}
\underline{x}:=\inf(\mathrm{dom}\,U).
\end{equation}

Let $U(+\infty):=\lim_{x\to +\infty}U(x)$ and define
\begin{equation}\label{xbar}
\xbar:=\inf\{x\mid U(x)=U(+\infty)\}.
\end{equation}
In the economic literature $\xbar$ is known as the \emph{satiation point} or
\emph{bliss point}. For strictly increasing utility functions $\xbar=+\infty$, while
for truncated utility functions, which feature for example in shortfall risk minimization,
$\xbar<+\infty$ represents a point where further increase in wealth does not produce
additional enjoyment in terms of utility. In economics this is interpreted as the
point of maximum satisfaction, or bliss.

By construction $\underline{x}\leq\xbar$ and the equality arises only when $U$
is constant on its entire effective domain in which case the utility maximization problem is trivial
since ``doing nothing'' is always optimal. Therefore, modulo a translation, the following assumption entails no loss of generality.
\begin{assumption}
 {$\underline{x}<0<\xbar$} and $U(0)=0$.
\end{assumption}

The convex conjugate of $U$ is defined by
$$ V(y) := \sup_{x\in\mathbb{R}} \{U(x) - xy \}.$$
  Our assumptions on $U$ imply that  $V$ is a proper, convex, lower semi-continuous function, equal to
$+\infty$ on $(-\infty, 0)$, and it verifies $V(0)=U(+\infty)$.
  For example, with exponential utility one obtains the following conjugate pair of functions $U,V$:
\begin{equation}\label{exponentialU}
U(x)=1-e^{- x}; \quad\quad V(y) = y\ln y - y + 1.
\end{equation}
In the sequel we will often exploit the following form of the Fenchel inequality, obtained as a
simple consequence of the definition of $V$:
\begin{equation}\label{fenchel}
U(x)\leq xy + V(y).
\end{equation}

\subsection{Young functions, Orlicz spaces and the Orlicz space induced by $U$}\label{Orlicz}
We recall basic facts on Young functions and induced Orlicz spaces. The  interested reader  is referred to the monographs by Rao and Ren \cite{RR} and Krasnose{\softl}skii and Rutickii \cite{KraRut61}  for proofs.

A Young function $\Psi: \mathbb{R}\rightarrow [0, +\infty]$ is an even, convex  and lower semicontinuous function with the properties:
$$ \text{(i)}\,\,\Psi(0)=0;\quad \text{(ii)}\,\,\Psi(+\infty)=+\infty;\quad \text{(iii)}\,\,\Psi<+\infty \text{ on an open neighborhood of } 0. $$
Note that $\Psi$ may jump to $+\infty$ outside a bounded neighborhood of $0$, but when $\Psi$ is finite valued,  it is also continuous by
convexity. In either case, $\Psi$ is  nondecreasing over $\mathbb{R}_+$ and \emph{countably convex} (see Lemma \ref{CC}).

The Orlicz space $L^\Psi $  induced  by $\Psi$ on
$(\Omega, \mathcal{F}_T, {P})$ is defined as
\begin{equation*}
{{L}^{\Psi }}=\{X \in L^{0}\mid  E[\Psi (c
X)]<+\infty \, \text{for some } c>0\}.
\end{equation*}
It is a Banach space when endowed with the Luxemburg (gauge) norm
\begin{equation*}
N_{\Psi }(X)=\inf \left\{ k >0\mid E\left[ \Psi \left( \frac{X}{k}\right) \right] \leq 1\right\}.
\end{equation*}
Orlicz spaces are generalizations of $L^p$ spaces whereby $\Psi(x) = |x|^p, p\geq1$ yields
$L^\Psi\equiv L^p,$ while $\Psi(x)=I_{\{|x|\leq1\}}$ induces the space $L^\infty$ with the supremum norm.
Intuitively, the faster $\Psi$ increases to $+\infty$ the smaller the space $L^\Psi$
and the stronger its topology. It is also clear that two \emph{distinct} choices of the Young
function may give rise to \emph{isomorphic} Orlicz spaces, the Luxemburg norms being equivalent. These statements
are made precise by the following definition and theorem.
\begin{definition}[Krasnose\softl skii and Rutickii]
Let $\Psi_1$ and $\Psi_2$ be two Young functions. We write $\Psi_1\succeq\Psi_2$, if there are constants $\lambda>0$ and $x_0$ such that for $x\geq x_0$,
$$ \Psi_1(\lambda x)\geq \Psi_2(x).$$ We say that $\Psi_1$ and $\Psi_2$ are equivalent if $\Psi_1\succeq\Psi_2$ and $\Psi_1\preceq\Psi_2$.
\end{definition}
\begin{theorem}[Krasnose\softl skii and Rutickii]
The following statements are equivalent:
\begin{romannum}
\item $\Psi_1\succeq\Psi_2$;
\item $L^{\Psi_1}\hookrightarrow  L^{\Psi_2}$;
\item there is $\lambda>0$ such that
$$ N_{\Psi_2}(X)\leq \lambda N_{\Psi_1}(X) \text{ for all } X\in L^{\Psi_1}.$$
\end{romannum}
\end{theorem}
\noindent Consequently, any Orlicz space $L^\Psi$ satisfies the embeddings
\begin{equation*}
{\ L^{\infty }\hookrightarrow {L}^{\Psi }\hookrightarrow L^{1}},
\end{equation*}
and two Orlicz spaces are isomorphic if and only if their Young functions are equivalent.

The \emph{Morse subspace} of $L^\Psi$, also  called the ``Orlicz heart'', is given by
$$M^\Psi = \{X\in L^0 \mid E[\Psi(c X )]<\infty \text{ for all } c>0 \}. $$
The inclusion of $M^\Psi$ in $L^{\Psi }$ may be strict and in particular $M^\Psi=\{0\}$ when $L^{\Psi}=L^{\infty}$.  On the other hand,
$M^p = L^p $ for any  $1\leq p< +\infty$.  More generally, when  $\Psi$ is finite on $\mathbb{R}$ then
\begin{equation}\label{inclOrlicz}
   L^{\infty}\hookrightarrow M^\Psi \hookrightarrow L^{\Psi}.
\end{equation}
We end these considerations with a classic  example of strict inclusion of $M^\Psi$ in $L^\Psi$.
\begin{example}\label{cosh}
Let $\Psi (x)= (\cosh x -1)$.  Simple calculations show that $L^{\Psi}$ is the space of random variables $X$ with \emph{some} absolute exponential moment finite, $E[e^{c |X|}]<+\infty $ for some $c>0$.  $M^{\Psi }$ is the proper subspace of those $X$ with \emph{all} absolute exponential moments finite. Therefore, as soon as $\Omega$ is infinite, $ M^{\Psi}\subsetneqq L^{\Psi}$.
\end{example}

From \S 3 onwards, the Young function will be
$$\Uhat(x):=-U(-|x|),$$
meaning that the Orlicz space in consideration is generated by the
lower tail of the utility function.  Then,
\begin{equation}\label{uhat-u-orlicz}
 X\in L^{\Uhat} \text{ iff } E[ U(-c |X|)]>-\infty
\text{ for some } c>0.
\end{equation}
For   utility functions with
lower tail which is asymptotically a power, say $p>1$, $L^\Uhat$ is isomorphic to  $ L^p $ and $L^\Uhat\equiv M^\Uhat$.
When $U$ is exponential,  say  $U(x) = 1-e^{-\gamma x}$, with $\gamma>0$,   $\Uhat(x) =
e^{\gamma|x|}-1 $    and  the induced space is isomorphic to that of Example \ref{cosh}, so that  $L^\Uhat\supsetneq
M^\Uhat $ in the relevant case $|\Omega|=+\infty$.   \\
 \indent For utility functions
with half-line as their effective domain, such as $U(x) =\ln (1+x)$,  $L^\Uhat$ is isomorphic to $L^{\infty}$ and $M^\Uhat=\{0\}$.

\subsection{Semimartingale norms}
There are two standard norms in stochastic
calculus. Let $S$ be an $\mathbb{R}^d$-valued  semimartingale  on
the filtered space $(\Omega, (\mathcal{F}_t)_{0\leq t\leq T}, P)$
and let $S^*_t =\sum_{i=1}^d \sup_{0\leq s\leq t } |S_s^i| $ be
the corresponding maximal process.
 For $p\in [1,\infty]$ let $$\|S\|_{\scr{S}^p}:=\|S^*_T\|_{L^p},$$ and
  denote the class of semimartingales  with finite $\scr{S}^p$-norm also by $\scr{S}^p$.
This definition is due to Meyer \cite{m78}. We extend the definition slightly to allow for an arbitrary Orlicz space
$L^\Psi(P)$ or its Morse subspace $M^\Psi(P)$,
\begin{align}
\scr{S}^\Psi &:= \{\text{semimartingale }S\mid S^*_T\in L^\Psi \},\label{scrS}\\
\SM &:= \{\text{semimartingale }S\mid S^*_T\in M^\Psi \}.\label{scrStil}
\end{align}

\begin{remark}\label{stopping}
Note for future use that $ \scr{S}^\Psi $ and $\SM$ are stable under
stopping, that is if $S$ belongs to $\scr{S}^\Psi$ or $\SM$ and if $\tau$ is a
stopping time, then the stopped process $S^\tau:= (S_{\tau\wedge t})_t$ is in
$\scr{S}^\Psi$ or $\SM$, respectively.
\end{remark}

Following Protter \cite{Pr05}, for any special
semimartingale $S$ with canonical decomposition into local
martingale part $M$ and predictable finite  variation part  $A$,
$S = S_{0}+M+A$, we define the following semimartingale norm,
$$\|S\|_{\scr{H}^{p}}=\|S_{0}\|_{L^{p}}+\|[M,M]_{T}^{1/2}\|_{L^{p}}+\|\mathrm{var}(A)_{T}\|_{L^{p}},$$
where $\mathrm{var}(A)$ denotes the absolute variation of process $A$.
 The class of processes with finite $\scr{H}^p$-norm is denoted by $\scr{H}^p$.
  As usual we let
$$ \scr{M}^p := \scr{H}^p\cap \scr{M}, $$ where $\scr{M}$ is the set of uniformly integrable $P$-martingales.

\subsection{Localization and beyond: $\sigma$-localization and $\I$-localization}
 Recall that  for a given
semimartingale $S$ on $(\Omega, (\mathcal{F}_t)_{0\leq t\leq T},
P)$,  $L(S)$ indicates the class of predictable and
$\mathbb{R}^d$-valued, $S$-integrable processes $H$ under
$P$, while $H\cdot S$ denotes the resulting scalar-valued integral process. In contrast,  when $\varphi$ is
a \emph{scalar} predictable process belonging to $\cap_{i=1}^d L(S^i)$ we follow \cite[\S IV.9]{Pr05} and \cite[Definition 8.3.2]{ds06} 
in writing $\varphi \cdot S$ for the \emph{vector-valued} process
$(\varphi \cdot S^1, \ldots, \varphi\cdot
S^d)$. \\
\indent  Now, let $\scr{C}$ be some fixed class of
semimartingales. The following methods of extending $\scr{C}$ appear in the literature:
\begin{enumerate}
\item[(i)]     $S\in\scr{C}_{\mathrm{loc}}$, i.e. $S$ is \emph{locally}
in $\scr{C}$, if there is a sequence of stopping times $\tau_n$
increasing to $+\infty$ (called \emph{localizing sequence}) such that each of the stopped processes
$S^{\tau_n}= I_{[0,\tau_n]}\cdot S$ is in $\scr{C}$.
 \item[(ii)]
  $S\in\scr{C}_{\sigma}$, i.e.  $S$ is $\sigma$-\emph{locally} in $\scr{C}$, if there is a sequence
of predictable sets $D_n$ increasing to $\Omega\times \mathbb{R}_+$ such that for every $n$  the vector-valued process  $I_{D_n}\cdot S$ is
in $\scr{C}$.
\item[(iii)]
  $S\in\scr{C}_{\I}$, i.e.  $S$ is $\I$-\emph{locally} in $\scr{C}$, if there is some scalar process
$\varphi\in \cap_{i=1}^d L(S^i)$, $\varphi>0$ such that $\varphi\cdot S$ is in $\scr{C}$.
\end{enumerate}

The first two items are standard (cf. \cite[I.1.33]{js03},
\cite{ka04}) while the third item  is an ad hoc definition. By
construction,
for an arbitrary semimartingale class
$\scr{C}$ one has $\scr{C}_{\sigma}\supseteq
\scr{C}_{\mathrm{loc}}\supseteq \scr{C}$. However it is not
\emph{a priori} clear what inclusions hold for $\scr{C}_{\I}$, apart
from the obvious $\scr{C}_{\I}\supseteq \scr{C}$. \'Emery
\cite[Proposition 2]{e80} has shown that when $\scr{C} =
\scr{M}^p$ or $\scr{H}^p $, the following equalities hold
\begin{equation}\label{emery}
 \scr{M}^p_\sigma = \scr{M}^p_\I,\quad\quad\scr{H}^p_\sigma = \scr{H}^p_\I, \quad\quad \text{for }p\in[1,+\infty).
\end{equation}
To complicate matters, some authors use $\sigma$-localization to mean $\I$-localization, see \cite{ds98,Pr05,ksi06}.
In this paper we deliberately make a \emph{clear distinction} between the two localization procedures.

The name $\I$-localization ($\I$ standing for integral) is
probably a misnomer, since no localization procedure is involved.
But we have chosen it because in \'Emery's  result $\I$-localization
coincides with $\sigma$-localization.  In general, however,
$\scr{C}_{\I}\neq \scr{C}_{\sigma}$. Intuition suggests that
the two localizations coincide  whenever the primary class
$\scr{C}$ is defined via some sort of integrability property, as
in the case above: martingale property and its generalizations,
boundedness or more generally Orlicz integrability conditions on
the maximal process. The next result in this direction  appears to
be new.
\begin{proposition}\label{sigmaI}
For any Orlicz space $L^\Psi$, its Morse subspace $M^\Psi$ and the
corresponding  semimartingale normed spaces $\scr{S}^\Psi,
\SM$, the following  identities hold:
$\SsigPsi=\SIPsi$ and
$\SsigMPsi=\SIMPsi$.
\end{proposition}

{\em Proof}. We prove the statement only for $\scr{S}^\Psi$, since the proof for
$\SM$ is analogous.
\begin{enumerate}
  \item[(i)]  Inclusion $ \SsigPsi\subseteq \SIPsi$. Fix $S\in \SsigPsi$. Then, there are predictable sets $D_n$ increasing to $\Omega\times \mathbb{R}_+$ such that
$(I_{D_n}\cdot S)^*_T\in L^\Psi$, for all $n\geq 1$. Thus there exist constants
$c_n >0$ such that $0\leq E[\Psi(c_n(I_{D_n}\cdot S)_T^*)]< +\infty$.   Since $\Psi$ is nondecreasing over $\mathbb{R}_+$,   $c_n$ can be assumed $(0,1]$-valued.
 Let
$$b_n:= E[\Psi(c_n(I_{D_n}\cdot S)_T^*)],  \quad\quad
d_n :=  h\, {2^{-n}} {(1+b_{n})^{-1}},$$
  where $
h:=1/(\sum_{n\geq 1} 2^{-n}(1+b_{n})^{-1})
$ is a normalizing constant,
and  define the following strictly positive, finite valued process $$  \varphi := \sum_{n\geq 1} c_n d_n I_{D_n}.   $$
Since  $0\leq \varphi_m:=\sum_{n=1}^m c_n d_n I_{D_n}\uparrow \varphi \leq \sum_{n\geq 1} d_n =1 $,    the Dominated Convergence Theorem for stochastic integrals (\cite[Theorem 32]{Pr05}) applies. Therefore,   $\varphi \in L(S)$ and $ (\varphi_m\cdot S  -\varphi\cdot S)^*_T  $ tends to $0$ in probability. Passing to a subsequence if necessary, we can assume the convergence holds $P$-a.s.  Now,
$$  (\varphi\cdot S)^*_T \leq (\varphi\cdot S - \varphi_m\cdot S)^*_T + (\varphi_m \cdot S)^*_T\leq   (\varphi\cdot S - \varphi_m\cdot S)^*_T + \sum_{n=1}^m c_n d_n (I_{D_n}\cdot S)^*_T  $$
and taking the limit on $m$,  $(\varphi\cdot S)^*_{T}\leq \sum_{n\geq 1} c_n d_n(I_{D_n}\cdot S)^*_{T} $. Monotonicity of $\Psi$ then ensures
$$
E[\Psi((\varphi\cdot S)^*_{T})]\leq E[\Psi(\sum_{n\geq 1}
c_n d_n(I_{D_n}\cdot S)^*_{T})]. $$
Countable convexity  of $\Psi$ (Lemma \ref{CC}) implies the  latter term  is majorized by $
\sum_{n\geq 1} d_n E[\Psi( c_n(I_{D_n}\cdot
S)^*_{T})]$ and thus
\begin{align*}
E[\Psi((\varphi\cdot S)^*_{T})]&\leq \sum_{n\geq 1} d_n E[\Psi( c_n(I_{D_n}\cdot
S)^*_{T})]\\
&= \sum_{n\geq 1}  d_n b_n = h \sum_{n\geq 1}  {2^{-n}} \frac{b_n}{ 1+b_{n} }\leq h \leq 2(1+b_1),
\end{align*}
i.e.  $S\in\SIPsi$.

  \item[(ii)]  Inclusion  $ \SIPsi\subseteq \SsigPsi$.  The line of the proof is:     
  (a) fix $S \in \SIPsi$ and show   $S\in(\scr{S}^\Psi_{\mathrm{loc}})_{\sigma}$; (b) then, as
$\scr{S}^\Psi$ is stable under stopping (see Remark \ref{stopping}), a result by Kallsen (\cite[Lemma 2.1]{ka04})
ensures $(\scr{S}^\Psi_{\mathrm{loc}})_{\sigma}= \SsigPsi$, whence the conclusion follows. \\
     We only need to prove (a), so let us fix $S \in \SIPsi$ and pick $\varphi>0$ such that
$\varphi \cdot S\in\scr{S}^\Psi$.  By construction $D_n := \{\frac{1}{n}<
\varphi< n\} $ is a sequence of predictable sets
increasing to $\Omega\times \mathbb{R}_+$. We now show $
I_{D_n}\cdot S\in\scr{S}^\Psi_\mathrm{loc}$ for all $n$. To this end,
let  $\tau^n_k =\inf\{ t\mid (I_{D_n}\cdot S)_t^*>k\} $. Then
\begin{align*}
(I_{D_n}\cdot S^{\tau^n_k})_T^*\, &\leq (I_{D_n}\cdot
S^{\tau^n_k})_{T-}^* + |  ( I_{D_n}\cdot S^{\tau^n_k})_T| \\
&\leq k +   | ( I_{D_n}\cdot S^{\tau^n_k})_{T-}| + |\Delta  ( I_{D_n}\cdot
 S^{\tau^n_k})_T|\leq 2k + |\Delta  ( I_{D_n}\cdot
 S^{\tau^n_k})_T|,
\end{align*}
and  the last jump term verifies
$$|\Delta  ( I_{D_n}\cdot S^{\tau^n_k})_T| =   |\Delta   (  \left(I_{D_n}/\varphi\right) \cdot ( \varphi \cdot
 S^{\tau^n_k}))_T| = \left(\frac{I_{D_n}}{\varphi}\right)_T |\Delta ( \varphi \cdot S^{\tau^n_k})_T| \leq n2(\varphi\cdot S)_T^*,
$$
so that
$$
  (I_{D_n}\cdot S^{\tau^n_k})_T^* \leq 2k   + 2n  (\varphi\cdot S)_T^*  \in
  L^\Psi.$$
  Therefore, for any fixed $n$,  $(I_{D_n}\cdot S^{\tau^n_k})_T^* $ is also in $L^{\Psi}$ for all $k$,
  whence  $ I_{D_n}\cdot S \in \scr{S}^\Psi_{\mathrm{loc}}$.
 This precisely  means  $S\in(\scr{S}^\Psi_{\mathrm{loc}})_{\sigma}$, which completes the proof.\qquad\endproof
\end{enumerate}

\section{The strategies}
\subsection{Conditions on $S$ and simple strategies}\label{sect: simple}
 Let  $S$ be  a
$d$-dimensional semimartingale which models the discounted
evolution of $d$ underlyings. As hinted in the introduction, to
accommodate popular models for $S$, including  exponential L\'evy processes,
we do not assume that $S$ is locally bounded. However, to make
sure that there is a sufficient number of well-behaved simple
strategies we impose the following condition on $S$:
\begin{assumption}\label{maximalS}
$S\in\SsigU$.
\end{assumption}

\noindent The class $\SsigU$ introduced here appears to be the most comprehensive class of price processes to have been systematically studied in the context of utility maximization to date. Most papers in the literature assume $S$ locally bounded, in our notation  $S\in\scr{S}_\mathrm{loc}^\infty$. Sigma-bounded semimartingales, that is processes in $\scr{S}\!\!_\sigma\!^\infty$, appear in Kramkov and S\^{i}rbu \cite{ksi06}.
 For  $p\in (1,+\infty)$ it can be shown, cf. \cite[Lemma A.2]{ck07}, that the class of semimartingales which are \emph{locally in} $L^p$ coincides with $\scr{S}_{\mathrm{loc}}^p$.
These processes feature in Delbaen and Schachermayer \cite{ds96}.   Biagini and Frittelli \cite{bf05} require existence of a suitable and compatible loss control for process $S$ which in our notation corresponds to $S\in\SIMU$. In \cite{bf08} this requirement is weakened to
$S\in \SIU$ which by Proposition \ref{sigmaI} is equivalent to Assumption \ref{maximalS}.

 As has already been pointed out in \cite{bf08}, the $\sigma$-localization in Assumption \ref{maximalS} provides a
substantial amount of flexibility since there are many interesting cases with $S\notin \scr{S}_{\mathrm{loc}}^\infty$ which fit in this setup.
However,  the cost of  considering  price processes of increasing generality is reflected
in progressively less attractive interpretations of simple trading strategies:
\begin{definition}\label{simple strategies}
Define $\varphi \in \cap_{i=1}^d L(S^i;P)$, $\varphi>0$,  and a   sequence of stopping times $(\tau_n)_n$ as follows:
\begin{enumerate}
\item[{\rm(i)}] For $S\in  \scr{S}^\Uhat$ let $\varphi\equiv 1$, $\tau_n\equiv T$ for all $n$;
\item[{\rm(ii)}] For $S\in  \scr{S}^\Uhat_\mathrm{loc}\setminus \scr{S}^\Uhat$ let $\varphi\equiv 1$ and let $(\tau_n)_n$ be  a localizing sequence for $S$ from the definition of $\scr{S}^\Uhat_\mathrm{loc}$;
\item[{\rm(iii)}] For $S\in\SsigU\setminus \scr{S}^{\Uhat}_{\mathrm{loc}}$  let $\tau_n\equiv T$ and let $\varphi$ be a fixed  $\mathcal{I}$-localizing integrand for $S$ such that $\varphi\,\cdot S\in \scr{S}^\Uhat$,  which is possible by virtue of Proposition \ref{sigmaI}.
\end{enumerate}
We say $H$ is a \emph{simple integrand} if it is of the form
$ H = \sum_{k=1}^N  H_k I_{ ]T_{k-1},T_k]}\,\varphi$
where  $ T_1\leq \cdots \leq  T_N$ is a finite sequence of  stopping times with $T_N$  dominated by $\tau_n$ for some $n$,  and  each $H_k$  is an $\mathbb{R}^d$-valued random variable,  $\mathcal{F}_{T_{k-1}}$-measurable and {bounded}.
The vector space of all simple integrands is denoted by $\simple$.
\end{definition}

As can be seen from the definition, when  $S\in\scr{S}^\Uhat$ no localization is needed. Every simple integrand is \emph{simple} also in the sense of integration theory and   it  represents a buy-and-hold strategy on $S$ over finitely many trading dates.  Vice versa, every buy-and-hold strategy implemented over a finite set of dates is simple. One may thus wonder which models  fall in this category. Some common examples are:
\begin{itemize}
  \item[(a)] discrete time models satisfying $|S_t|\in L^{\Uhat}$ for $t=1,2,\ldots,T$;
  \item[(b)]   L\'evy processes, when (i) the utility $U$ is exponential and the L\'{e}vy measure $\nu$ satisfies
  $$\int e^{\lambda|x|} I_{\{ |x|>1\}} d\nu(x)<+\infty {\text{ for some } \lambda>0};$$
  or (ii)  the utility  $U(x)$ behaves asymptotically like $-|x|^p,p>1$ when $x\rightarrow -\infty$ and the L\'{e}vy measure $\nu$ satisfies
  $$\int |x|^p I_{\{ |x|>1\}} d\nu(x)<+\infty.$$
  Such conditions on $\nu$ are equivalent to integrability conditions on the maximal functional $S^*$, i.e.  $ S\in\scr{S}^\Uhat$, which in turn are equivalent to $\Uhat$-integrability of $S_t$ at \emph{some } $t>0$.   This follows from general results on $g$-moments of L\'{e}vy processes, when $g$ is a submultiplicative function  (see \cite[Theorems 25.3 and 25.18]{Sat99}). Explicit examples of utility maximization in this case can be found in Biagini and Frittelli \cite[\S 3.2]{bf05}, \cite[Example 35]{bf08}. Here, $U$ is exponential utility and  $S$ is a compound Poisson process with Gaussian or doubly exponentially distributed jumps;
  \item[(c)]   exponential L\'evy processes belong to $\scr{S}^{\Uhat}$ whenever $\Uhat$ behaves asymptotically like a power function with exponent $p\in(1,+\infty)$ and the L\'evy measure of $\ln S$, $\nu$, satisfies
  $$\int e^{px} I_{\{ x>1\}} d\nu(x)<+\infty.$$
  This is derived similarly as in (b) once $\ln S$ has been decomposed into a sum of two independent L\'evy processes, one of which represents large jumps of $\ln S$.
\end{itemize}

 For $S\in  \scr{S}^\Uhat_\mathrm{loc}\setminus \scr{S}^\Uhat$ it is still true that all simple strategies are of the buy-and-hold type but one can no longer pick the trading dates arbitrarily.   From a practical point of view most commonly used price processes fall into this category. For example, in the
Black-Scholes model the risky asset is represented by a geometric Brownian motion which does not belong to  $\scr{S}^\Uhat$ when
$U$ stands for the exponential utility. On the other hand $S$ is
continuous and therefore locally bounded which means
$S\in\scr{S}^\infty_\mathrm{loc}\subseteq
\scr{S}^\Uhat_\mathrm{loc}\subseteq \SsigU$ for any
utility function satisfying our  assumptions, including the
exponential. The same line of reasoning applies to diffusions
and more generally to all semimartingales with bounded jumps which
therefore automatically belong to $\SsigU$ for \emph{any} utility function $U$. In the special case $\scr{S}^\Uhat_\mathrm{loc} = \scr{S}^p_\mathrm{loc}$ our definition of simple strategies mirrors the definition in Delbaen and Schachermayer \cite{ds96}.

Finally, the price paid for
 allowing $S\in\SsigU\setminus \scr{S}^{\Uhat}_{\mathrm{loc}}$ is that simple strategies can no longer
be interpreted as buy-and-hold with respect to the original price process $S$ but only with respect to the better-behaved process
$S':=\varphi\cdot S$.
This case is interesting mainly theoretically since the $\I$-localizing strategy $\varphi$ has already appeared in the literature on utility maximization. It plays an important role in the work of Biagini \cite{b04} where  the maximal process $(\varphi\cdot S)^*$ is taken as a dynamic loss control for the strategies in the utility maximization problem. Within setups of increasing generality in Biagini and Frittelli \cite{bf05,bf08} $\varphi$ gives rise to so-called
\emph{suitable} and \emph{(weakly) compatible} loss control variables $W:= (\varphi\cdot S)^*_T$.
\subsection{$\sigma$-martingale measures}
To motivate the definition of simple strategies mathematically we now define dual asset pricing measures.
\begin{definition}\label{sigma-mart}
$Q\ll P$ is a $\sigma$-martingale measure  for $S$ iff $S$ is a
$\sigma$-martingale under $Q$.  The set of all
$\sigma$-martingales measures for $S$ is denoted by  $
\mathcal{M}$ and the subset of equivalent measures by $
\mathcal{M}^e$.
\end{definition}

The concept of $\sigma$-martingale measure was introduced to
Mathematical Finance  by Delbaen and Schachermayer \cite{ds98}.
When $S$ is (locally) bounded, it can be shown that $\mathcal{M}$
coincides with  the absolutely continuous (local) martingale
measures for $S$ (see e.g. Protter \cite[Theorem 91]{Pr05}).
Therefore,
 $\sigma$-martingales are a natural generalization of local
martingales in the case when $S$ is not locally bounded and the
elements of $\mathcal{M}$ which are equivalent to $P$ can be used
as  arbitrage-free pricing measures  for the derivative securities whose payoff depends on $S$.
The recent book \cite{ds06} contains  an extensive
treatment of the financial applications of this mathematical
concept.

When $S\in\SsigU\setminus \scr{S}^{\Uhat} $, one may wonder to what extent the utility maximization problem depends on the particular choice of $\varphi$ (or of the localizing sequence $(\tau_n)_n$).  Thanks to \'Emery's equality \eqref{emery}
the set of absolutely continuous $\sigma$-martingale measures for
$S$ is the same as the set of $\sigma$-martingale measures for
$S'=\varphi\cdot S$. Specifically, $Q\ll P$ is a $\sigma$-martingale
measure for $S$ by \eqref{emery} if and only if  there exists a
$Q$-positive, predictable process $\psi_Q \in \cap_{i=1}^d
L(S^i;Q)$ such that $ \psi_Q \cdot S$ is a $Q$-martingale. And
this happens if and only if  $ \psi'_Q\cdot (\varphi \cdot S) $ is
a $Q$-martingale, where $\psi'_Q = \frac{\psi_Q}{\varphi}$.

Since the sets of $\sigma$-martingale measures for $S$ and
$S'$ are the same,  the \emph{dual} problem to the utility maximization also remains the same.  Under suitable conditions (see the statement of the main Theorem \ref{main theorem}), we thus end up with the same optimizer, regardless of
a specific choice of the $\I$-localizing strategy $\varphi$.

\subsection{Generalized relative entropy and properties of simple integrals}
\begin{definition}\label{P-V}
A probability $Q$ has finite  generalized  relative entropy with respect to $P$, notation: $Q\in
P_V$, if there is $y_Q>0$ such that
\begin{equation}\label{fin-entropy}
v_Q(y_Q):=E\bigg[ V\bigg( y_Q\frac{dQ}{dP} \bigg) \bigg]<\infty.
\end{equation}
\end{definition}
For exponential utility $U(x)= 1-e^{-x}$ we have seen in \eqref{exponentialU} that $V(y)= y\ln y -y+1$, and in this case a probability $Q$ verifies \eqref{fin-entropy} if and only if its probability density has finite Kullback-Leibler \cite{KuLei51} divergence:
$$ H(Q\| P):=E\left[\frac{dQ}{dP}\ln \frac{dQ}{dP}\right]<+\infty.$$
The Kullback-Leibler divergence is also known in Information Theory as \emph{relative entropy} of $Q$ with respect to $P$.
Intuitively speaking, $H(Q\|P)$ is a non-symmetric measure of the distance between probabilities $Q$ and $P$. In Financial Economics it measures the extra amount of wealth an agent with exponential utility perceives to have if she invests optimally in a complete market with pricing measure $Q$, as opposed to investing all her wealth in the risk-free asset.

In the 1960-ies, Csisz\'{a}r treated a wide class of statistical distances replacing the weighting  function $y\ln y$ by a convex function $V$ verifying $V(1)=0$. In his terminology, $Q$ has finite $V$-divergence with respect to $P$ if
\begin{equation}\label{Vdiv}
    E\left[V\left(\frac{dQ}{dP}\right)\right]<+\infty.
\end{equation}
The interested reader can also consult Liese and Vajda \cite{lv87}.

In Mathematical Finance applications the function  $V$ is typically the convex conjugate of a utility function, see Kramkov and Schachermayer \cite{ks99}, Bellini and Frittelli
\cite{bef}, Goll and R\"{u}schendorf \cite{gr01} and basically all
the contemporary literature on utility maximization. Here,  a $Q\in P_V$ is said to have \emph{finite generalized relative entropy}. Our  definition pushes the generalization one step further, since we do not require $y_Q=1$ in \eqref{fin-entropy}.

The proof of the following simple Lemma is omitted.
\begin{lemma}\label{P_V convexity}
Consider $Q_i\ll P$, $i=1,2$, such that $v_{Q_i}(y_i)<+\infty$ for some $y_i>0$. Then for $0\leq\lambda\leq 1$
$$v_{\lambda Q_1+(1-\lambda)Q_2}\left( \frac{1}{\lambda/y_1+(1-\lambda)/y_2}\right)<\infty.$$
\end{lemma}
\begin{corollary}
$P_V$ is convex.\\
\end{corollary}
Simple integrals have good mathematical properties
with respect to $\sigma$-martingale measures with   finite   generalized  relative
entropy.
\begin{lemma}\label{simplesuper}
The wealth process $X = H \cdot S$ of every
 $H\in \simple$ is a uniformly integrable martingale under all $Q\in \measures$.
\end{lemma}
\begin{proof}
 (i) $S\in\SsigU\setminus \scr{S}^{\Uhat}_{\mathrm{loc}}$.
  Since $H\in \simple$,  the maximal functional $X^*$ verifies
$ X^*_T \leq c  (\varphi \cdot
S)^*_T $ for some constant $c>0$ and some  $\I$-localizing
integrand  $\varphi$ which exists by  Proposition \ref{sigmaI}. By \eqref{uhat-u-orlicz} then  $E[
U(-\alpha (\varphi \cdot S)^*_T )] \in \mathbb{R}$ for some constant
$\alpha>0$ and,  as a consequence,
$$ 0\geq  E\bigg[U\bigg ( - \frac{\alpha}{c} X^*_T  \bigg ) \bigg]
>-\infty.
$$
For any fixed $Q\in \measures$, the  Fenchel inequality  $U(x) -xy
\leq V(y) $ applied with $ x = - \frac{\alpha}{c} X^*_T, y = y_Q
\frac{dQ}{dP} $ gives
$$  U\bigg( - \frac{\alpha}{c} X^*_T   \bigg)   +  \frac{\alpha}{c} X^*_T y_Q \frac{dQ}{dP} \leq V\bigg( y_Q\frac{dQ}{dP} \bigg),  $$
whence $  0\leq \frac{\alpha}{c} y_Q\,  X^*_T  \frac{dQ}{dP} \leq
V( y_Q\frac{dQ}{dP})  - U( - \frac{\alpha}{c} X^*_T  ), $ and
therefore $X^*_T$ is in $L^1(Q)$.  As $Q$ is a $\sigma$-martingale
probability for $S$,  $X$ is also a $Q$-$\sigma$-martingale. Since
its maximal process is integrable, $X$ is in fact a $Q$-uniformly
integrable martingale (see Protter \cite[Chapter IV-9]{Pr05}).\\
\indent (ii) $S\in \scr{S}^{\Uhat}_{\mathrm{loc}}$. Proceed
as in (i), replacing $\varphi$ with $I_{[0,\tau_n]}$.\qquad
\end{proof}

In financial terms,  the message of the above Lemma  is that each
$Q\in\measures$ represents a pricing rule that assigns a correct
price to every simple self-financing strategy.

\subsection{Admissible integrands and integrals}
As anticipated in the introduction, simple integrands are unlikely to contain the solution of the utility maximization problem.  The appropriate class of admissible integrands is an extension given in terms of suitable limits of strategies in $\simple$. We recall the definition of admissibility here for convenience.
\newcounter{auxS}
\newcounter{auxT}
\setcounter{auxS}{\value{section}}
\setcounter{auxT}{\value{theorem}}
\setcounter{section}{1}
\setcounter{theorem}{0}
\begin{definition}
  $H\in L(S)$ is an \emph{admissible integrand} if  $ U (H\cdot S_T) \in L^1(P)$ and if
  there exists an approximating sequence $(H^n)_n$ in
 $\simple $ such that:
\begin{romannum}
    \item   $H^n\cdot S_t \rightarrow  H\cdot S_t $ in probability
    for all $t \in [0,T]$;
    \item  $ U( H^n\cdot S_T )  \rightarrow  U( H\cdot S_T ) $ in
    $L^1(P)$.
\end{romannum}
The set of all admissible integrands is denoted  by $\adm$.
\end{definition}
\setcounter{section}{\value{auxS}}
\setcounter{theorem}{\value{auxT}}

While for $H \in \simple$ the wealth process $H\cdot S$ is always
a martingale under $Q\in
\measures$ due to Lemma \ref{simplesuper}, the following result shows that $\adm$ is a subset of
the supermartingale class  of strategies $\super$ introduced by
\cite{Sch03},
\begin{equation}\label{Hsuper}
	\begin{split}
		\super := \{ H \in L(S)
		\mid H\cdot S \text{ is } &\text{a local martingale} \\
		\text{and }&\text{a supermartingale under any } Q\in
		\measures\}.
	\end{split}
\end{equation}

\begin{proposition} \label{admsuper} $\adm\subseteq\super$.
\end{proposition}
\begin{proof}
Let  $X=H\cdot S$ for some $H\in\adm $ and let $ (X^n:=H^n\cdot
S)_n$ with $H^n\in\simple$ be an  approximating sequence.
Fix a $Q\in \measures$ and a corresponding scaling
$y_Q$ as in Definition \ref{P-V}. Item (i) of Definition \ref{adm} applied at time $T$
implies $( X^n_T)^-$ converges in $P$-probability
to $X_T^-$. Moreover,  Fenchel inequality gives
 \[  U(X^n_T )- V(y_Q\frac{dQ}{dP}) \leq  X^n_T \,y_Q \frac{dQ}{dP}.  \]
 From Definition \ref{adm}, item (ii), the left hand side above converges in $L^1(P)$,   whence
  the family  $(Y^n)_n, Y^n:= (X^n_T)^-\frac{dQ}{dP}$ is
$P$-uniformly integrable, so $ ((X^n_T)^-)_n$ is $Q$-uniformly
integrable (see Lemma \ref{ui}).  Uniform integrability plus convergence in probability
ensures $(X^n_T)^-\rightarrow X^{-}_T$ in $L^1(Q)$.
By   passing to a subsequence if necessary,  the next is an integrable lower bound for $(X^n_T)_n$, 
$$ W^Q := \sum_n |(X^{n+1}_T)^- -(X^n_T)^- | \in L^1(Q).$$ 
Denote by $Z^Q$ the associated
$Q$-martingale, $Z^Q_t:=E_Q[ W^Q \mid \mathcal{F}_t]$.  Note that when
$\mathrm{dom}\,U$ is a half-line we could also have chosen
trivially $W^Q:=-\inf\mathrm{dom}\,U$.

Since $ X^n_T \geq -W^Q$
and  process $X^n $ is a $Q$-martingale for all $n$ by Lemma
\ref{simplesuper},  we obtain
\begin{equation}\label{Wcontr}
X^n_t =  E_Q[X^n_T \mid \mathcal{F}_t] \geq -E_Q[W^Q \mid
\mathcal{F}_t]=-Z^Q_t,
\end{equation}
so that the sequence $X^n$ is controlled from below by the
$Q$-martingale $Z^Q$. Therefore by Delbaen and Schachermayer
compactness result \cite[Theorem D]{ds99} (in the version stated
in \S 5, \cite{ds98}) there exists a limit c\`{a}dl\`{a}g
supermartingale $\widetilde{V}$ to which a sequence  $K^n\cdot S$,
where  $K^n$ is a suitable  convex combination of tails $K^n \in
\mathrm{conv} (  H^n  , H^{n+1} , \ldots ) $, converges $Q$-almost surely
for every rational time $0\leq q\leq T$. By item (i),
 $(X_t^n)_n$ converges in $P$-probability to $X_t$ for every $t$, thus
 $K^n \cdot S_t$
 converges to $X_t$  for every $t$ as well.
  Therefore $\widetilde{V}$
coincides $Q$-a.s. with $X$ on rational times, and since $X$ is
also c\`{a}dl\`{a}g as it is an integral, $X$ and $\widetilde{V} $
are indistinguishable, so  that  $X$  is   a $Q$-supermartingale.
By assumption $Q$ is a $\sigma$-martingale measure, so  $X =H\cdot
S=(\frac{1}{\varphi} H)\cdot(\varphi\cdot S)$ where $\varphi>0$
and $\varphi\cdot S$ is a $Q$-martingale. As $X$ also satisfies $X\geq -Z^Q$,  Ansel
and Stricker lemma \cite[Corollaire 3.5]{as94} implies that
 $X$ is a local $Q$-martingale.\qquad
\end{proof}

\begin{remark}
Proposition \ref{admsuper} would go through if one replaced our class $\simple$ with
the set of integrands with wealth bounded from below
\begin{equation}\label{bounded below}
\bounded =\{H\in L(S)\mid H\cdot S\geq c \text{ for some }c\in\mathbb{R}\},
\end{equation}
as in Schachermayer \cite{Sch01} when $S\in\scr{S}^\infty_{\mathrm{loc}}$,
or more generally with the larger set of strategies whose losses are in
some sense well controlled as in Biagini and Frittelli
\cite{bf05,bf08},
\begin{equation}\label{boundedU}
\boundedU =\{H\in L(S)\mid \exists W\geq 0,   E[U(-W)]>-\infty, H\cdot S\geq- W\},
\end{equation}
 see also Biagini and S\^{i}rbu \cite{bs09}. An
application of the Ansel and Stricker lemma \cite[Corollaire
3.5]{as94} shows that wealth processes for strategies in
$\boundedU\supseteq\bounded$ are  local martingales and supermartingales  under any
$Q\in \measures $ -- but not
 martingales in general. In contrast, our smaller class $\simple$ has the stronger martingale property as shown in Lemma \ref{simplesuper}. Mathematically, however, it is the supermartingale property of approximating strategies that really matters.
 This is also true in the proof of the main Theorem \ref{main theorem} where one can replace arguments relying on the martingale property
 of approximating strategies \cite[Corollaire 2.5.2]{y78} with supermartingale compactness results of \cite{ds99}.
\end{remark}

\bigskip

The    list below summarizes the
advantages of $\adm$ over current definitions of admissibility:
\begin{itemize}
    \item[(a)] Definition \ref{adm} is \emph{primal}. No pricing measures come into play,
         and admissibility can thus be checked under $P$.
    \item[(b)] The present definition is \emph{dynamic}, that is the whole wealth process, rather than just its terminal value, is involved in
         the definition of $\adm$. As a result \emph{all} admissible strategies are in the supermartingale class.
    \item[(c)] The loss controls required in the proof of the supermartingale property are generated endogenously, via
      approximating sequences. This provides a great deal of flexibility and  ensures that for $U$ finite on $\mathbb{R}$
        the optimizer is in $\adm$ under very mild conditions, milder than the conditions assumed
        to obtain the supermartingale property of the optimizer in \cite{Sch03,bf07}. Since under
        our assumptions the optimal utilities over $\adm$ and $\super$ coincide, see \eqref{eqndual2}, the smaller class $\adm$
        seems to be more appropriate than $\super$ not only economically but also mathematically.
    \item[(d)] Approximation by   strategies in $\simple$ is built into the definition
         of admissibility, it does not have to be deduced separately (cf.
         \cite{str03}).
    \item[(e)] The desirable properties above hold without any technical assumptions on
        $U$. It can be finite on $\mathbb{R}$ or only on a half-line; bounded from above or not,
        or even truncated; neither strict monotonicity, strict convexity nor differentiability are required.
    \item[(f)] Our definition is compatible with the existing definition of admissibility for
        non-monotone quadratic preferences, see Remark \ref{rem: quadratic U} below. We have therefore
        found a good  notion of admissibility which encompasses both the classical mean-variance
        preferences \emph{and} monotone expected utility.
\end{itemize}

\begin{remark} \label{rem: quadratic U} For the purpose of this remark only,  we admit non-monotone $U$.
  Specifically,
let $U(x):=x-x^2/2$, which represents a normalized quadratic
utility.
In such case, $H\in\adm$ if and only if there is a sequence of
$H^n\in\simple$ such that: 
\begin{itemize}
\item[{\rm(a)}] $H^n\cdot S_t\rightarrow H\cdot S_t$
in probability for all $t\in[0,T]$, and 
\item[{\rm(b)}] $H^n\cdot S_T\rightarrow H\cdot S_T$ in $L^2(P)$. 
\end{itemize}
In other words, when $U$ is quadratic the
admissibility criterion in Definition \ref{adm} coincides with the
notion of admissibility pioneered by Jan   Kallsen in
\cite[Definition 2.2]{ck07}, which inspired our work.  
\end{remark}

{\em Proof}. Since (a) above  and (i) in Definition
\ref{adm} coincide, the only thing to prove is that (ii) in our
definition  is equivalent to (b) above:
\begin{itemize}
    \item[$\Rightarrow$]  Suppose first  $H \in \adm$.
        The $L^1(P)$ convergence of utilities implies $E[U(X^n_T)]\rightarrow E[U(X_T)]$
        so that  $X^n_T$ are uniformly bounded in $L^2(P)$. Since $L^2(P)$ is a reflexive space
        there is a sequence of convex combinations of tails
        $(X_T^k)_{k\geq n}$, say $\widetilde{X}^n_T$,
        which converges in $L^2(P)$   to a square integrable random variable which necessarily is
        $X_T =H\cdot S_T$ thanks to Definition \ref{adm}, item (i).
         By considering the corresponding convex combinations
         of strategies, which are again simple, we obtain the existence of an approximating
         sequence \`{a} la Kallsen for $H$.
  \item[$\Leftarrow$] Conversely, let  $X = H\cdot S$ be  an
        integral  approximated \`{a} la Kallsen
        by simple integrals $(X^n)_n$.  $L^1(P)$ convergence of the utilities
        $U(X^n_T)$ to $U(X_T)$ is then a consequence of
        the Cauchy-Schwartz inequality.\qquad\endproof
        \end{itemize}


\section{Optimal trading strategy is in $\adm$}\label{sect: umax} The
optimal investment problem can be formulated over $\simple$,  $\adm$ or
over $\super$, respectively,
\begin{align}
u_{\simple}(x)&:=\sup_{H\in \simple} E[{U(x+ H\cdot S_T)}],\label{utmaxsimple}\\
u_{\adm}(x)&:=\sup_{H\in \adm} E[{U(x+ H\cdot S_T)}],\label{utmax}\\
u_{\super}(x)&:=\sup_{H\in \super} E[{U(x+ H\cdot S_T)}].\label{utmaxsuper}
\end{align}
\noindent Alongside, we consider   auxiliary complete market
utility maximization problems,  each obtained by fixing an arbitrary  $Q\in \measures$:
\begin{equation}\label{uQdef}
  {u}_Q(x) : = \sup_{X\in L^1(Q), E_Q[X]\leq x } E[  {U}(X)].
\end{equation}

The value functions  $u_{\simple}(x), u_{\adm}(x), u_{\super}(x), u_Q(x)$ are
also known as indirect utilities
 (from the respective domains of
maximization). The next lemma is an easy consequence of the
definition of $\adm$
and of the supermartingale property of the strategies in $\adm$ and $\super$.
The proof is omitted.
\begin{lemma}\label{ineq}
For any $x> \underline{x}$ and for any $Q\in \measures$
\begin{equation}\label{ineq2}
u_{\simple}(x) =  u_{\adm}(x) \leq u_{\super}(x) \leq u_{Q}(x).
\end{equation}
\end{lemma}
\subsection{Reasonable Asymptotic Elasticity and Inada conditions}
It is well known in the literature that  the  existence of an
optimizer is not guaranteed yet,  neither  in  $\adm$ nor in the larger
supermartingale class  $\super\supseteq\adm$.
An additional condition has to be imposed, essentially to ensure
that the expected utility functional $k \mapsto E[U(k)]$ is
upper semicontinuous with respect to  some weak topology on
terminal wealths.

Kramkov and Schachermayer  were the first to address
 this  issue in \cite{ks99,Sch01} for regular $U$, that is utilities that are strictly increasing,
 strictly concave and differentiable in the interior of their effective domain.
  To the  end of recovering an  optimizer  they introduced the celebrated
  Reasonable Asymptotic Elasticity  condition on $U$ ($\mathrm{RAE}(U)$),
\begin{align}
&\limsup_{x\rightarrow +\infty} \frac{x U'(x)}{U(x)}<1,\label{rae+}\\
\text{and also} \ \  &
\liminf_{x\rightarrow \,-\infty} \frac{xU'(x)}{U(x)}>1,\text{ when
}  U \text{ is finite on } \mathbb{R},\label{rae-}
\end{align}
as a necessary and sufficient condition to be imposed on the utility
$U$ \emph{only}, regardless of the probabilistic model. This
condition is  now very popular, see \cite{oz09,rs05,Sch03,bou}
to mention just a few contributions.

In subsequent work, in the context of utilities finite on $\mathbb{R}_+$, Kramkov and
Schachermayer \cite{ks03} put forward less restrictive conditions%
\footnote[1]{The interested reader is referred also to the recent
Biagini and Guasoni \cite{bg09} for counterexamples and a
different, \emph{relaxed} framework that allows optimal terminal
wealth to be a measure and not only a random variable.}\!\!,
\emph{imposed jointly on the model and on the preferences}, in
order to recover the optimal terminal wealth. Here they  work
under assumptions which are equivalent to the existence of $Q\in\emeasures$
and the following Inada condition on
the indirect utility $u_{\bounded}$, where the class $\bounded$ is
defined in \eqref{bounded below}:
\begin{equation}\label{Inada3}
 \lim_{x\rightarrow +\infty} u_{\bounded}(x)/x=0.
\end{equation}
It is important to note that for utility functions finite on a half-line the modulus of the conjugate function $V(y)$ grows only linearly for large $y$
and therefore the following implication holds automatically:
\begin{equation}\label{Inada5}
Q\in\measures\Rightarrow v_Q(y)<+\infty \text{ for all } y \text{ sufficiently high.}
\end{equation}
On the other hand, for utilities finite on $\mathbb{R}$ condition \eqref{Inada5} has to be imposed explicitly, together with an appropriate generalization of condition (\ref{Inada3}).
\begin{assumption}\label{Inada} Condition \eqref{Inada5} is satisfied and
\begin{equation}\label{Inada2}
\text{there exists } Q\in\measures \text{ such that } \lim_{x\rightarrow +\infty} u_Q(x)/x  =0.
\end{equation}
\end{assumption}
\noindent Note first that the requirement \eqref{Inada2} automatically holds, and for all $Q\in \measures$, if $U(+\infty)<+\infty$.
Since for any $Q\in \measures$ one has $u_Q(x)\geq
u_{\simple}(x)\geq U(x)$, and $U$ is monotone, condition
\eqref{Inada2} implies an identical Inada condition both on the
indirect utility $u_\simple$ and also on the original utility function $U$
at $+\infty$.  An identical chain of inequalities for the indirect utilities holds if we replace $\simple$ with $\bounded$ and for this reason condition
\eqref{Inada2} is slightly stronger than the condition
$\eqref{Inada3}$ imposed in \cite{ks03}
when $U$ is finite on a half-line. It is an open question whether condition \eqref{Inada2} can be weakened to
\begin{equation}\label{Inada6}
\measures\neq\emptyset \text{ and } \lim_{x\rightarrow +\infty} u_\simple(x)/x  =0.
\end{equation}
Further discussion of Assumption \ref{Inada} and its relation to $\mathrm{RAE}(U)$  and
the Inada condition \eqref{Inada3} can be found in \S \ref{sect_Inada}. The results of the next section go in that direction.

\subsection{Complete market duality}
Here we study a complete market $Q\in P_V$ and hence no specific model
for $S$ is required.  Among other results, we provide an
alternative characterization of the Inada condition \eqref{Inada2}
in terms of the generalized relative entropy of $Q$.

\begin{lemma}\label{dualityQ}
 Fix $Q\in P_V$ and consider the function $v_Q$ defined in \eqref{fin-entropy}. For any $x> \underline{x}$,
     \begin{equation}\label{uQdual1}
        u_Q(x) = \min_{y\geq 0}\{ xy + v_Q(y)\}<+\infty.
     \end{equation}
\end{lemma}
\noindent An Orlicz duality based proof of the above lemma is given in \S \ref{auxiliary}. Here we only remark that the  minimizer may not be unique. This is due to lack of strict convexity of $V$, which in turn is due to lack of strict concavity of $U$.
\begin{corollary} \label{CorInada}Fix $Q\in P_V$. The following statements are equivalent:
\begin{enumerate}
\item[{\rm(i)}] $u_Q$ verifies the Inada condition at $+\infty$: $\lim_{x\to
+\infty}u_Q(x)/x = 0$; 
\item[{\rm(ii)}] there is $y_Q>0 $ such that
\begin{equation}\label{InadaEntropy}
v_Q(y)=E\bigg[V\bigg(y\frac{dQ}{dP} \bigg) \bigg]<+\infty \text{ for all } y \in (0,y_Q].
\end{equation}
\end{enumerate}
\end{corollary}
{\em Proof}.
(ii)$\Rightarrow$(i)
    Suppose that
    $v_Q(y)$ is finite in a right neighborhood of $0$. By Fenchel
    inequality, $ E[U(X)] -E[y\frac{dQ}{dP}X] \leq E[V(
    y\frac{dQ}{dP})] $ for all $X\in L^1(Q)$ so that $u_Q(x)
        \leq xy + v_Q(y) $ for all
    $y>0$. Fixing  $y$ one obtains
    $\lim_{x\rightarrow +\infty} u_Q(x)/x \leq y   $ and on letting
    $y\rightarrow 0$ the Inada condition on $u_Q$ follows.
    
 (i)$\Rightarrow$(ii) For a given $x> \underline{x}$,   select one
    dual minimizer in \eqref{uQdual1} and denote it by $y_x$.
    Now, $u_Q(x) = xy_x +v_Q(y_x) $, $v_Q(y_x)$ is finite,
          and  the chain of inequalities  $$ u_Q(x)= xy_x +
        v_Q(y_x)
        \stackrel{Jensen}{\geq} x y_x + V(y_x)\geq x y_x  \geq  0
        $$
       holds  for any $x$ as $V$ is nonnegative.  Dividing by $x>0$ and sending $x$
        to $+\infty$, \eqref{Inada2}  implies $\lim_{x\rightarrow
        +\infty }  y_x   =0$. Finiteness of $v_Q$ over the   set $\{y_x\}_x$, whose closure contains
        $0$,
           and  convexity of $v_Q$ finally imply  $v_Q$ is
        finite in the  interval $(0, y_{Q}]$, with $y_Q$ from \eqref{fin-entropy}.\qquad\endproof

\begin{corollary}\label{EqInada}
If $\emeasures\neq\emptyset$ then the measure $Q$ in \eqref{Inada2} can be chosen equivalent to $P$.
\end{corollary}
\begin{proof}
Take $Q^e\in\emeasures$ and assume $Q$ satisfies \eqref{Inada2}. By Corollary \ref{CorInada} $v_Q(y)$ is finite for all $y$ near zero.
Define $$Q^*:=\frac{1}{2}Q+\frac{1}{2}Q^e.$$ Thus, $Q^*\sim P$ and  by Lemma \ref{P_V convexity} $v_{Q^*}(y)$ is finite for all $y$ near zero. Therefore  $u_{Q^*}$ satisfies the Inada conditon \eqref{Inada2}. \qquad
\end{proof}

The next proposition contains a novel characterization of the
condition $$ u_Q(x)< U(+\infty),$$  which is a kind of ``no
utility-based arbitrage'' condition,  when $Q$ has finite generalized relative entropy.
Agents cannot reach satiation utility $U(+\infty)$ if the initial
capital  $x$ is below the satiation point $\xbar$, and vice versa.

\begin{proposition}\label{uQstrict}
For $Q\in P_V$ and $x>\underline{x}$ the following statements are equivalent:
\begin{romannum}
    \item $ x<\xbar$;
    \item
    \begin{equation}\label{uQdual2}
        u_Q(x) = \min_{y>0}\{  xy + v_Q(y)\}<U(+\infty).
        \end{equation}
\end{romannum}
\end{proposition}
\begin{proof}
 (ii)$\Rightarrow$(i) $ U(x)\leq u_Q(x)< U(+\infty)$ implies $x<\xbar$.  \\
 (i)$\Rightarrow$(ii)  Let
    $Z:=y_QdQ/dP$, with $y_Q$ from \eqref{fin-entropy}. When
    $U(+\infty)=V(0)=+\infty$ there is nothing to prove in view of
    \eqref{uQdual1}.
    Consider therefore the remaining case $0<U(+\infty)=V(0)<+\infty$. Function $f(y):=V(y)+xy$ is convex and by Rockafellar \cite[Theorem 23.5]{r70}
it attains its minimum at $\opty:=U_{-}^\prime(x)>0$ with
$f(\opty)=V(\opty)+x\opty=U(x)$.
    Convexity then gives
    \begin{align*}
    f(y)&\leq f(0)-y\frac{f(0)-f(\opty)}{\opty}=U(+\infty)-y\frac
    {U(+\infty)-U(x)}{\opty}&\text{ for } y\in [0,\opty],\\
    f(y/k)&\leq f(0)+\frac{f(y)-f(0)}{k}\leq
    {U(+\infty)}+\frac{f(y)}{k}&\text{ for } k\geq 1,y\geq 0.
    \end{align*}
    For $k\geq 1$ these estimates imply
    \begin{align*}
    E[f(Z/k)]&=E[f(Z/k)1_{\{ Z\leq k\opty \} }]+E[f(Z/k)1_{\{Z>k\opty\}}]\\
    &\leq U(+\infty) {- \frac{1}{k}
     \bigg(\frac{U(+\infty)-U(x)}{\opty}E[Z1_{\{ Z\leq k\opty\}}]- E[f(Z)1_{\{Z>k\opty\}}] \bigg)},
    \end{align*}
   and, as  $x< \xbar$ implies $U(x)<U(+\infty)$,
   for sufficiently large $k$   $E[f(Z/k)]< U(+\infty)=V(0)$, which    completes the
    proof.\qquad
\end{proof}
\begin{remark}
Corollary \ref{CorInada} and Proposition \ref{uQstrict} should be contrasted with an example by
Schachermayer \cite[Lemma 3.8]{Sch01}, where the author constructs an
arbitrage-free complete market with unique  pricing measure $Q$
for which $u_Q(x)\equiv U(+\infty)$, while $U$ is strictly increasing and bounded from above
(and therefore it satisfies the Inada condition at $+\infty$). This is possible because
the measure $Q$ in question does not belong to $P_V$.
\end{remark}
\begin{corollary}
If $x \in (\underline{x}, \xbar)$ then
\begin{equation}\label{dualineq}
u_{\simple}(x)  \leq u_{\adm}(x)\leq u_{\super}(x) \leq \inf_{Q\in \measures  }u_Q(x) = \inf_{ y> 0, Q\in \measures} \{ xy + v_Q(y) \}.
\end{equation}
\begin{proof}
The inequalities follow from Lemma \ref{ineq} and Proposition \ref{uQstrict}.\qquad
\end{proof}
\end{corollary}

\subsection{The main result}
The minimization problem on the right-hand side of \eqref{dualineq} is a natural candidate as 
a dual problem to the utility maximization on the left-hand side.
However,  the general theory of \cite{bf08} shows  that  in order to catch the minimizer the 
dual domain must be extended beyond probability densities.  Rephrased in our terminology,  
whenever $S\in\SsigU\setminus\SsigMU$ 
the dual problem may have a  minimizer which has a non zero singular part, 
{but  for $S\in \SsigMU$ the singular parts in the dual problem 
disappear and there is no duality gap in \eqref{dualineq} under Assumption \ref{Inada}. 
We make these statements precise in Theorem \ref{generaldual} and Corollary \ref{corstrongdual}. 

Our main result hinges on the absence of singularities in the dual problem, which is what we now assume.
Within the confines of Assumption \ref{strongdual}, which  can be imposed also
when $S\in\SsigU\setminus\SsigMU$, we provide a unified treatment for utility functions finite on $\mathbb{R}$ or only on a half-line.
\begin{assumption}\label{strongdual}
For any   {$x \in (\underline{x}, \xbar)$}, the following dual relation holds:
\begin{equation}\label{eqndual}
 u_{\simple}(x)  = \min_{Q\in \measures  }u_Q(x) = \min_{ y\geq 0, Q\in \measures} \{ xy + v_Q(y) \}.
\end{equation}
\end{assumption}

\noindent As indicated above, this assumption represents no loss of generality for $S\in \SsigMU,$ including situations where $U$ is finite on $\mathbb{R}$ and
              \begin{itemize}
               \item[(a)] $S$ is ``sufficiently integrable". Some commonly found examples are locally bounded processes,  such as diffusions or jump diffusions with bounded relative jumps,  regardless of the specification of $U$; jump diffusions  with   relative jumps in $M^{\Uhat}$;  L\'{e}vy processes with
               large jumps in $M^{\Uhat}$;
                \item[(b)] $L^{\Uhat}=M^{\Uhat}$, under the standing Assumption \ref{maximalS}. This happens when e.g.  $U$ has left tail that  behaves asymptotically like a power, $x^p$, with $p>1$.
              \end{itemize}
When $S\in\SsigU\setminus\SsigMU$, which includes all cases where $U$ is finite only on a half-line, unfortunately there is no known sufficient condition for the strong duality \eqref{eqndual} to hold. The appropriate modification of Theorem \ref{main theorem} which would work without Assumption \ref{strongdual}
 remains an interesting area for future research.}

Assumption \ref{strongdual} together with    \eqref{dualineq} immediately yields the following,  apparently stronger, statement for   $x \in (\underline{x}, \xbar)$
\begin{equation}\label{eqndual2}
u_{\simple}(x) = u_{\adm}(x) = u_{\super}(x)= \min_{y>0, Q \in \measures }\bigg\{  xy +  E\bigg[ V\bigg(y\frac{dQ}{dP} \bigg)\bigg]\bigg\} .
\end{equation}
Any  optimal  dual  pair in \eqref{eqndual2} is denoted by  $(\opty,\optQ)$, dependence on $x$ is understood.
The lack of uniqueness of the optimal dual pair is again due to the lack of strict convexity
of $V$, stemming from the lack of strict concavity of $U$.

Most results in the literature are obtained under the assumption $\optQ\sim P$. This condition is satisfied automatically for utility functions unbounded from above
since $V(0)=U(+\infty) = +\infty$ while $E[V(\opty\frac{d\optQ}{dP})] $ must be finite.
When $U$ is strictly monotone  but bounded, a well-known sufficient condition for $\optQ\sim P$ is the existence of an equivalent $\sigma$-martingale measure with finite generalized relative entropy.
This can be gleaned from (a.i) and (a.iii) in Theorem \ref{main theorem},  on observing that $\xbar=+\infty$.

As a general  comment,    Theorem \ref{main theorem}   provides  a
desirable approximation result for the optimal strategy $\optH\in\adm$.
The approximation holds under very mild
conditions:  $U$ may lack strict monotonicity and strict
concavity; $S\in \SsigU$; and $\optQ$ may be only absolutely continuous with respect to $P$. These results are novel not only for utility finite on $\mathbb{R}$ but also for utility functions finite on
a half-line.

For $U$ finite on $\mathbb{R}$ our framework is a further improvement over the current literature:
\cite{Sch01}, \cite{kabstr02}, \cite{str03}, \cite{oz09},
\cite{btz04} all assume $S$ locally bounded.
Approximation by simple
strategies has so far been shown only for exponential utility,
 for locally bounded $S$ and for expected
utility only cf. \cite[Theorem 5]{str03} -- not
in the stronger sense of $L^1(P)$ convergence of the utilities
given by item (ii) in Definition \ref{adm}.

For comparison, Schachermayer \cite{Sch01} proves an approximation
similar to \eqref{convU} for the terminal wealth of the optimal solution $\optf =
\optH \cdot S_T$ via integrals
bounded from below.
This work is extended further by Bouchard et al. \cite{btz04} who allow for non-differentiable and non-monotone utility functions.
Moreover, in \cite{Sch03} $\optH$ is shown to be in the supermartingale class of strategies through a (hard) contradiction argument, which
is later extended by \cite{bf07} to   $S\in \SIMU$  with a proof along the same lines.

In the present paper  the supermartingale property of $\optH $  is shown in a general setup and in
 a very natural  way, as a consequence of $\adm\subseteq\super$.  We also extend results of Bouchard et al.
\cite{btz04} beyond $S\in\scr{S}^\infty_\mathrm{loc}$ under the weaker condition
from Assumption \ref{Inada} instead of the $\mathrm{RAE}(U)$ condition \eqref{RAE2},
while considerably simplifying the required proofs thanks to the
Orlicz duality approach.

When $U$ is not strictly monotone, that is when $U$ attains its global maximum at a satiation point $\xbar<+\infty$,
 the  sufficient conditions for $\optQ\sim P$ known in
the monotone case do not work; here typically $\optQ $ is
\emph{not} equivalent to $P$ even when there are equivalent
probabilities in $\measures$. We nonetheless recover an integral representation \emph{under $P$}, and thus existence of an optimal trading strategy,  provided
the budget constraint is binding, $E_Q[\optf]=x$, for \emph{some} $Q\in\emeasures$. This mild sufficient condition appears to be new in the literature.
  Our contribution in the case where $U$ is strictly monotone but $\optQ$ is not equivalent to $P$
is discussed in detail in \S \ref{literature2}.

\begin{theorem}\label{main theorem}
Under Assumptions \ref{maximalS}, \ref{Inada} and \ref{strongdual}, for any initial wealth 
$x  \in (\underline{x}, \xbar)$ the following statements hold:
\begin{enumerate}
\item[{\em (a)}]    There exists a $(-\infty,+\infty]$-valued claim $\optf$, not unique in
    general, with the following properties
      \begin{enumerate}
      \item[{\rm (i)}] $\optf<+\infty$ whenever $\emeasures\neq\emptyset$;
      \item[{\rm (ii)}] $\optf$ realizes the optimal expected utility, in the sense that $$  E[U({\optf})] = u_{\simple}(x);$$
      \item[{\rm (iii)}]  $E_{\optQ}[ \optf] = x$, and the following equalities hold $P$-a.s. for any  dual optimizers  $\opty,\optQ$:
      \begin{align*}
       V\bigg(\opty \frac{d\optQ}{dP} \bigg)&=U(\optf) - \optf \opty\frac{d\optQ}{dP},\\
       \{ \optf \geq \xbar \} &=  \bigg\{\frac{d\optQ}{dP } = 0\bigg\};
       \end{align*}
       \item[{\rm (iv)}] $ \optf\in L^1(Q) \text{ and } E_Q[\optf]\leq x \text{ for all } Q\in\measures $;
      \item[{\rm (v)}] In case $U$ is strictly concave, $V$ is strictly convex and the solutions of primal and dual problem $\optf, \opty, \optQ $ are unique.
            If in addition $U$ is differentiable,  these unique solutions  satisfy $\opty \frac{  d\optQ}{dP} =U'(\optf) $;
      \end{enumerate}
\item[{\em (b)}] There is an approximating sequence of   strategies $H^n \in \simple$ with terminal values $f^n := x+H^n\cdot S_T$ such that:
      \begin{enumerate}
      \item[{\rm (i)}]
      \begin{equation} \label{Pas convergence}
      f_n\stackrel{P\text{-}a.s.}{\rightarrow}\optf,
      \end{equation}
      provided $\emeasures\neq\emptyset$ or $\xbar=+\infty$;
      \item[{\rm (ii)}]
        \begin{equation}\label{convU}
            U(f_n) \stackrel{L^1(P)}{\rightarrow} U(\optf);
        \end{equation}
      \item[{\rm (iii)}]
      \begin{equation}\label{L1(optQ)conv}
        f_n \stackrel{L^1(\optQ)}{\rightarrow} \optf,
        \end{equation}
      and, provided \eqref{Pas convergence} holds, for any $Q\in\measures$ such that $E_Q[\optf]=x$
        \begin{equation}\label{L1(Q)conv}
        f_n \stackrel{L^1(Q)}{\rightarrow} \optf;
      	\end{equation}
      \item[{\rm (iv)}]  There exists an integral representation $ \optf = x+ \optH\cdot S_T$
        with $\optH\in L(S;\optQ)$, and $\optH\cdot S$ is a $\optQ$-martingale.
      \item[{\rm (v)}]  When there is $\widetilde{Q}\in\emeasures$ such that $E_{\widetilde{Q}}[\optf]=x$ then $\optH$ in {\rm(iv)}
          can be chosen in $\adm$ and consequently $\optH$ is a  utility maximizer over both $\adm$ and $\super$,
        \begin{equation}\label{main result}
         u_{\simple}(x) = u_{\super}(x) = \max_{ H \in \adm }E[U(x +H\cdot S_T)].
         \end{equation}
				In particular, by virtue of {\rm(a.iii)}, \eqref{main result} holds whenever $\optQ\sim P$.
        \end{enumerate}
\end{enumerate}
\end{theorem}

{\em Proof}.
\begin{enumerate}
\item[(a)] Let us fix a pair $ \opty, \optQ$  of  dual minimizers. For ease of notation and without loss of generality we let $x=0$ throughout.
 \begin{enumerate}
  \item[(i.1)]   Select a maximizing sequence $(k_n)_n, k_n= K^n\cdot S_T, K^n \in
               \simple $ so that $ E[U(k_n)] \uparrow u_\simple(0)  $. Fix $Q^*\in \measures $ as follows:
                \begin{itemize}
                  \item in case $\emeasures\neq\emptyset$,  select $Q^*$ as an equivalent  measure satisfying  \eqref{Inada2}. This is possible by Corollary \ref{EqInada};
                  \item in case $\emeasures=\emptyset$,   take  $Q^*=\optQ$. Here   necessarily  $V(0)=U(+\infty)<+\infty$, so $\optQ$ as well as any other measure in $\measures$ satisfies  \eqref{Inada2}.
                \end{itemize}
                Let  $$\overline{Q}:=\frac{1}{2}\optQ + \frac{1}{2}Q^*.$$
                 Then, $\overline{Q}\in \measures$; $\overline{Q}\sim P$ if $\emeasures \neq \emptyset$; $\overline{Q}$ satisfies \eqref{Inada2}; $ L^1(\overline{Q})=L^1(\optQ)\cap L^1(Q^*) $; and $L^1(\overline{Q})$-convergence
                is equivalent to  convergence in $L^1(\optQ)$ \emph{and }$L^1(Q^*)$ by construction.

    \item[(i.2)] The sequence $(k_n)_n$ is bounded in $L^1(\overline{Q})$. In a general case this follows from the auxiliary Proposition \ref{LQ-boundedness}, which  in turn is a consequence of the Inada condition \eqref{Inada2}. In a special case when $\mathrm{dom}\,U$ is a
            half-line,  $L^1(\overline{Q})$-boundedness also follows trivially from
            $k_n \geq \underline{x}$ and $E_Q[k_n]= 0$, which is a consequence of Lemma \ref{simplesuper}.
            In a second special case where $U$ is bounded from above the claim can be alternatively deduced
            from the boundedness of $U^-(f_n)$ and the Fenchel inequality \eqref{fenchel}.
    \item[(i.3)]   $L^1(\overline{Q})$ boundedness of $(k_n)_n$ enables the application of  the Koml\'{o}s theorem, so that
            there exists a  sequence of convex combinations $(f_n)_n$ with $f_n \in \mathrm{conv}(k_n, k_{n+1}, \ldots)$, that
            converges  $\overline{Q}$-a.s. to a certain random variable {$f\in L^1(\overline{Q})\subseteq L^1(\optQ)$.}
            As $\simple$ is a vector space,  these $f_n$ are  terminal values of simple integrals $f_n = H^n\cdot S_T$, $H^n\in\simple$.
              By concavity,  the  $f_n$ are still maximizers, i.e. $E[U(f_n)] \uparrow u_\simple(0)$.
    \item[(i.4)] Define $\optf$ as follows:
             \begin{itemize}
               \item in case $\emeasures\neq\emptyset$, $\optf:=f$. Here, $\overline{Q} \sim P$ and $f$ is a well-defined element of $L^0(\Omega, \mathcal{F}_T, P)$ with $f_n\stackrel{P\text{-}a.s.}{\rightarrow}f=\optf$;
               \item in case $\emeasures=\emptyset$ , and thus $\overline{Q}=\optQ $,
             $$ \optf:=f I_{\{ \frac{d\optQ}{dP} > 0\}} +  \xbar I_{\{ \frac{d\optQ}{dP} =0 \}}   $$
             \end{itemize}
             By construction,  $\optf\in   L^1(\optQ)$ in both cases.
  \item[(ii)]   It is easily seen that {for $y>0$} and $Q\in \{ \optQ, \overline{Q}\}$
             \begin{equation} \label{aux}
             \limsup_n\bigg( U( f_n) - f_n { y} \frac{dQ}{dP} \bigg)\leq U(\optf)-\optf { y} \frac{d  Q}{dP}\leq V\bigg (y\frac{d Q}{dP} \bigg ) ,
             \end{equation}
             using the convention $+\infty \cdot 0=0$.
            The Fatou lemma applied to \eqref{aux} for any $y$ sufficiently large  yields
        		\begin{align}
        			u_{\simple}(0)&= \limsup_n E\bigg[ U( f_n) - f_n { y} \frac{dQ}{dP} \bigg]
         				\leq E\bigg[\limsup_n\bigg( U( f_n) - f_n { y} \frac{dQ}{dP} \bigg) \bigg]\nonumber\\
            		&\leq  E\bigg[ U(\optf) - \optf{y} \frac{dQ}{dP} \bigg]\leq E\left[ V\bigg (y\frac{d Q}{dP} \bigg )\right].\label{limsup}
         		\end{align}
  In particular, we derive  $U(\optf)\in L^1(P)$.
   			On taking $Q=\overline{Q}$, in virtue of \eqref{Inada2} and Corollary \ref{CorInada} we can  let  $y\to 0$ to obtain
        \begin{equation} \label{auxfhat1}
        		u_{\simple}(0)\leq E[U(\optf)].
        \end{equation}
        		Also, on taking  $Q=\optQ$
        and sending $y\to+\infty$ we get
        \begin{equation}\label{auxfhat2}
        E_{\optQ}[\optf]\leq 0.
        \end{equation}
        Equation \eqref{limsup} with the choice of the optimizers $Q=\optQ,y=\opty$ yields
        \begin{equation}\label{ineq-eq}
          u_{\simple}(0) \leq  E[ U(\optf)] - \opty E_{\optQ}[\optf]\leq E\bigg[V\bigg( \opty \frac{d\optQ}{dP}\bigg)\bigg]=u_{\simple}(0),
         \end{equation}
        which implies $E_{\optQ}[\optf]=0$ and $u_{\simple}(0)= E[U(\optf)]$, in view of \eqref{auxfhat1}, \eqref{auxfhat2} and  $\opty>0$ from \eqref{eqndual2}.
    \item[(iii)] The Fenchel optimal relation $  U(\optf) -\optf  \opty \frac{d\optQ}{dP} \stackrel{P\text{-a.s.}}{=} V( \opty\frac{d\optQ}{dP})
        $ now follows from (\ref{auxfhat1}-\ref{ineq-eq}). From here we conclude
        $$ \frac{d\optQ}{dP}=0 \Leftrightarrow U(\optf)=U(\xbar)=U(+\infty).$$ The forward implication follows from $V(0)=U(+\infty)$ and the converse from $\opty>0$. The equality $E_{\optQ}[\optf] = 0 $ has just been shown in (a.ii.2).
    \item[(iv)] Since $\limsup_n( U( f_n) -
         f_n { y} \frac{d\optQ}{dP})\leq  U( \optf) - \optf { y} \frac{d\optQ}{dP}
         $ and the inequalities in \eqref{limsup} are
         equalities for $y=\opty$ and $Q=\optQ$, one has  $\limsup_n U(f_n)
         =U(\optf)=U(+\infty)$ on $A:=\{\frac{d\optQ}{dP}=0\}$.
         Therefore, by passing to  a subsequence that converges to the
         limsup we can assume $U(f_n)I_{A} \rightarrow U(\optf)I_{A} $, whence globally
        \begin{equation}\label{U Pas convergence}
        U(f_n) \stackrel{P\text{-a.s.}}{\rightarrow} U(\optf).
        \end{equation}
       Consider now an arbitrary $Q\in \measures$.
        Given \eqref{U Pas convergence} necessarily $\liminf_n
        f_n I_{A} \geq \xbar I_{A}$ and therefore 
        $$\liminf_n |f_n|\geq|\optf|,\text{ and } \liminf_n f_n\geq\optf.$$  
        Additionally, $(f_n)_n $
        is $L^1(Q)$ bounded:  $E_Q[f_n]=0$ and $(E[U(f_n)])_n$ is bounded from below, so  Proposition
        \ref{LQ-boundedness}  applies again. Therefore, Fatou Lemma yields  $\optf\in L^1(Q)$  and
   			  $$E_Q[\optf]\leq E_Q[ \liminf_n f_n] \leq \liminf_n E_Q[f_n]=0.$$

     \item[(v)]
        Finally,   the results when $U$ is strictly
        concave and differentiable follow now  from the pointwise  identity $ U( x) - x U'(x) =V(U'(x))$.
\end{enumerate}
\item[(b)]

\begin{enumerate}

    \item[(i)] This follows by construction when $\emeasures\neq\emptyset$, cf. item (a.i.3) above, and otherwise from $U(f_n)\rightarrow U(\optf)$ when $\xbar=+\infty$, cf. equation \eqref{U Pas convergence}.
    \item[(ii)]      Since  $U(f_n) \stackrel{P\text{-a.s.}}{\rightarrow}
        U(\optf) $, the $L^1$ convergence of the utilities is equivalent
        to showing uniform integrability of $ (U(f_n))_n$. Given the
        convergence of the expected utility, $E[U(f_n)] \uparrow  E[U(  \optf
        )]$, an  argument ``\`{a} la Scheff\'{e}'' shows
        that the  uniform integrability of $ (U(f_n))_n$ is equivalent to uniform
        integrability of any of the two families $ (U^-(f_n))_n $, $ (U^+(f_n))_n$.
        $U(0)=0$ and  monotonicity of $U$    imply
        $U^-(f_n)= -U(-f_n^- )$ and   $U^+(f_n) = U(f_n^+)$.

        Suppose by contradiction that the family $(U^+(f_n))_n \equiv (U(f^+_n))_n$ is \emph{not}
        uniformly   integrable, and proceed as in \cite[Lemma 1]{ks03}.  Given the supposed lack of uniform integrability,
        there exist disjoint measurable
        sets $(A_n)_n$  and a   constant  $\alpha >0$ such that
        $$ E[U( f_n^+) I_{A_n}] \geq  \alpha. $$
        Set $g_n = \sum_{i=1}^n f_i^+ I_{A_i}$ and fix a $Q\in \measures$ satisfying the Inada condition \eqref{Inada2}.
         $(f_n)_n$ is $L^1(Q)$ bounded by Proposition \ref{LQ-boundedness}
         and  clearly $ E_{Q}[g_n]\leq  n C $ where $C$ is a positive   bound on the $L^1(Q)$ norms of the
        sequence  $(f_n)_n$.   In addition, $ E[U(g_n)] \geq  n \alpha
        $ because the $(A_n)_n$ are disjoint.  Therefore,
        $$ \frac{u_{{Q}}( n C )}{nC} \geq  \frac{E[ U(g_n) ]}{nC} \geq  \frac{\alpha}{C}>0  $$
        and passing to the limit  when $n \uparrow \infty$ the conclusion  contradicts \eqref{Inada2}.
        So the family $(U^+(f_n))_n$ is uniformly integrable, and   $(U(f_n))_n$ as well, which means   $U(f_n)
        $ tends
        in $L^1(P)$ to $U(\optf)$.
    \item[(iii)]   To see that $ f_n \rightarrow
        \optf$ in $L^1(\optQ)$, from
        $  U(\optf)- \optf \opty  \frac{d\optQ}{dP} =
        V(\opty\frac{d\optQ}{dP})
        \geq    U(f_n) - f_n \opty \frac{d\optQ}{dP}  $
         the difference $  U(\optf) -U(f_n)   - (\optf-f_n) \opty  \frac{d\optQ}{dP}
        $
        is nonnegative and has  $P$-expectation which tends to zero. Henceforth such  difference
        is $L^1(P)$ convergent
        to $0$, which, thanks to $L^1(P)$ convergence of $
        U(f)-U(f_n)$,
        yields $L^1(P)$ convergence  to $0$ of $ (\optf-f_n)
        \frac{d\optQ}{dP}$.\vspace{0.3cm}\\
        From Fenchel inequality,
        $$ f_n^- y_Q \frac{dQ}{dP} \leq V \bigg(y_Q \frac{dQ}{dP}\bigg) - U(-f_n^-)\leq V\bigg(y_Q \frac{dQ}{dP}\bigg)+ | U(f_n)| $$
        and given the $P$-uniform integrability of $(U(f_n))_n$, proved in (b.ii), the $Q$-uniform integrability
        of $ (f_n^-)_n$ follows (see Lemma \ref{ui}).
        Admitting $f_n\stackrel{P\text{-a.s.}}{\rightarrow}\optf$ and $E_Q[\optf]=0$,
        and in view of $0 =\lim_n E_Q[f_n]$, an application of the Scheff\'e lemma again  yields $f_n\stackrel{L^1(Q)}{\rightarrow}\optf$.
    \item[(iv)]
    		Recall that  $X^n:= H^n \cdot S$ are all $\optQ$ \emph{uniformly integrable} martingales  by Lemma
         \ref{simplesuper}. Moreover,   $\optQ$  is a $\sigma$-martingale measure for $S$, so $ X^n  =
         (H^n \frac{1}{\varphi_{\optQ}}) \cdot (\varphi_{\optQ}\cdot S)$, where
        $M = \varphi_{\optQ} \cdot S$ is a $\optQ$ martingale and $\varphi_{\optQ}>0$ holds $\optQ$-a.s.
        The convergence \eqref{L1(Q)conv} permits a straightforward application of a celebrated result
        by Yor \cite{y78} on the closure of stochastic integrals,  which gives an integral
        representation with respect to $M$   of the limit $\optf$  under $\optQ$,
        $ \optf = H^* \cdot M_T =  \optH \cdot S_T $, with $\optH = H^* \varphi_{\optQ}$,
          and the optimal process $\widehat{X}:= \optH \cdot S $
        is also  a $\optQ$-uniformly integrable martingale.
    \item[(v)]
        When there is $\widetilde{Q}\in\emeasures$ with $E_{\widetilde{Q}}[\optf]=0$  convergence \eqref{Pas convergence} applies and by virtue of (b.iii)
        the construction of $\optH$ can be performed under
        $\widetilde{Q}$ instead of $\optQ$ and therefore $\optH\in L(S,P)$.
        To show $\optH \in \adm$, note we have already proved  \eqref{convU} so we only need
        convergence in $P$-probability of the wealth process at intermediate times.
        The convergence in \eqref{L1(Q)conv} and the martingale
        property of the $X^n$ and of $ \optH\cdot S$ under $\widetilde{Q}$ imply
         $$ E_{\widetilde{Q}}[ |X^n_t- \optH\cdot S_t|  ] = E_{\widetilde{Q}}
          [ |E_{\widetilde{Q}}[ X^n_T- \optH\cdot S_T  \mid \mathcal{F}_t]| ]
         \stackrel{\text{Jensen}}{\leq} E_{\widetilde{Q}}[| X^n_T- \optH\cdot S_T|]. $$
        Therefore,  for any $t$,  $X^n_t \rightarrow  \optH\cdot
        S_t$ in $L^1(\widetilde{Q})$ and therefore in $\widetilde{Q}$-probability,
        which is equivalent to convergence in $P$-probability.
         Thus,  $\optH\in  \adm$ follows.\qquad\endproof

\end{enumerate}
\end{enumerate}


\section{On the main  assumptions and connections to literature}\label{literature}
\subsection{More details on Assumption \ref{Inada}}\label{sect_Inada}
Condition \eqref{Inada5} is automatically satisfied for utilities finite on a half-line.
For utilities finite on $\mathbb{R}$ it makes sure that  the claim $\optf$ constructed via the Koml\'{o}s theorem
satisfies the budget constraint $E_Q[\optf]\leq x$ for every $Q\in\measures$.

To the best of our knowledge  Assumption
\ref{Inada} is strictly weaker than any other assumption used in
the current literature for $U$ finite on $\mathbb{R}$. In current
references, the typical assumption is
$\mathrm{RAE}(U)$, which implies
$v_Q(y)<+\infty$ for all $ y>0$ and for all $Q\in P_V$ by
\cite[Corollary 4.2]{Sch01}, whence Assumption \ref{Inada}
necessarily holds. In the non-smooth utility case studied by
Bouchard et al. \cite{btz04}, equivalent asymptotic elasticity
conditions are imposed on the Fenchel conjugate $V$,
\begin{equation} \label{RAE2}
\lim_{y\to 0_+}\frac{|V^\prime_{-}(y)|y}{V(y)}<+\infty,\quad\quad \lim_{y\to +\infty}\frac{|V^\prime_{+}(y)|y}{V(y)}<+\infty.
\end{equation}
These again imply $v_Q(y)<+\infty$ for all $ y>0$ and for all $Q\in P_V$, see \cite[Lemma 2.3]{btz04}.

 On the other hand,  Biagini and Frittelli \cite{bf05,bf08}  do not require $\mathrm{RAE}(U)$, but
  instead assume that $v_Q(y)$ is finite for all $Q \in \measures $ and all $y>0$,  which is  weaker than $\mathrm{RAE}(U)$
  but clearly stronger than Assumption \ref{Inada} by virtue of Corollary \ref{CorInada}.
 Since condition \eqref{Inada2} is only slightly stronger than the truly necessary condition \eqref{Inada3}
 for utility functions finite on a half-line, Assumption \ref{Inada} seems to be a very good choice for a unified treatment of utility
maximization problems, regardless of the domain of $U$.

\subsection{A general duality formula and more details on Assumption \ref{strongdual}}\label{sect_generalduality}
 Duality theory applied in the   Orlicz spaces context  shows   that the dual
problem associated with the utility maximization over a general Orlicz
space may contain singular parts, see \cite{bf08}. We have tried to make this section as self-contained as possible, but the reader can find more details on the structure of the dual of a general Orlicz space in  \cite{RR}.
    The  dual variables $z \in  (L^{\Uhat})^* $ have, in general, a two-way
decomposition $z = z_r +z_s$ in regular and singular part, where $z_r$ only can be identified
with a measure absolutely continuous with respect to $P$. Let   $\langle \cdot, \cdot\rangle$  denote the bilinear form
for the dual system $( L^{\Uhat},(L^{\Uhat})^* )$. The convex conjugate $(I_U)^*:(L^{\Uhat})^* \rightarrow (-\infty, +\infty]$ of the expected utility functional  $ L^{\Uhat} \ni k \mapsto E[U(k)]:= I_U(k) $ is then defined as:
$$  (I_U)^*(z):= \sup_{k \in L^{\Uhat} } \{ I_U(k)- \langle z,k \rangle\}. $$
Recall that the polar set of a cone   $C\subset L^{\Uhat}$ is the subset of $(L^{\Uhat})^*  $  defined as $C^0:= \{ z \in (L^{\Uhat})^* \mid
 \langle z, k \rangle \leq 0 \text{ for all } k \in C\}$. The set of normalized elements in $C^0$, i.e. those $z$ which verify $\langle z, I_{\Omega}\rangle=1$, is denoted by $C^0_1$. Thus,  when $z\in C^0_1 $ is regular it is an absolutely continuous normalized measure (with sign).   The following Theorem  is the key to  understanding the precise implications of Assumption \ref{strongdual}. Its proof is basically identical to \cite[Theorem 21]{bf08}, but with  our strategies  $\simple$.
\begin{theorem}\label{generaldual}
Under Assumption  \ref{maximalS} and \ref{Inada}, for any $x\in (\underline{x},\xbar)$
the following dual relation holds:
 \begin{equation}
 u_{\simple}(x)  =  \min_{z \in C^0}  (I_U)^*(z)= \min_{ y>0 , z\in C^0_1}
 \left\{ y(x+\|z_s\|) + E\left[V\left(y \frac{dz_r}{dP}\right)\right]  \right\},\label{generaldual2}
\end{equation}
where $C:=\{ k \in L^{\Uhat} \mid k \leq H\cdot S_T \text{ for  some } H \in \simple \}$.
  When there is a regular dual minimizer, the above formula simplifies to
   \begin{equation}\label{strongdual2}
     u_{\simple}(x)  =   \min_{ y>0 , Q\in \measures} \left\{ yx + E\left[V\left(y \frac{dQ}{dP}\right)\right]   \right\}.
   \end{equation}
\end{theorem}
\begin{proof}
 The first part of the  proof goes along the same lines of the proof of Lemma  \ref{dualityQ} and thus we give only a sketch.  Suppose for simplicity $x=0$.    As in Lemma  \ref{dualityQ}, $u_{\simple}(0) = \sup_{k \in C} E[U(k)]$ and the  concave expected utility functional $I_U$ is proper and has a continuity point which belongs to $C$. Then, Fenchel Duality Theorem applies and
\begin{equation}\label{desp}
u_{\simple}(0) =\sup_{k \in C} E[U(k)] = \min_{z\in C^0} (I_U)^*(z) = \min_{z\in C^0} \left\{ E\left[V\left( \frac{dz_r}{dP}\right)\right] + \|z_s\|\right\},
\end{equation}
where the second equality follows from the explicit expression of the convex conjugate $(I_U)^*(z)=E[V( \frac{dz_r}{dP})] + \|z_s\| $  found by Kozek \cite{koz}. Note that $C\supseteq -L^{\Uhat}_+ $, so $C^0_1$ consists of positive normalized functionals. Assumption \ref{Inada}
 implies in particular $ \measures \neq \emptyset$ and since  $0\in (\underline{x}, \xbar)$  Proposition
 \ref{uQstrict} implies $ u_{\simple}(0) \leq u_Q(0) < U(+\infty)
 $ for any $Q\in \measures$.
Thus $u_{\simple}(0)<U(+\infty) $, so the dual minimizers are non null and  the dual problem can be re-written as
$$\min_{y>0, z\in C_1^0} \left\{ E\left[V\left( y \frac{dz_r}{dP}\right)\right] + y \|z_s\|\right\}, $$
via the normalized dual
variables in $C^0_1$, which proves \eqref{generaldual2}.   Any dual minimizer  $\hat{z}\in C^0_1$ clearly satisfies the integrability condition $ E[V(y \frac{d\hat{z}_r}{dP})] <+\infty $ for some $y$. Since $\langle \hat{z} , I_ \Omega \rangle = E[\frac{d\hat{z}_r}{dP}I_{\Omega}] +\langle \hat{z}_s, I_{\Omega}\rangle =1 $, when $\hat{z}_s=0$  this exactly means $\hat{z}=\hat{z}_r \in P_V$.  Suppose there exists  a regular  dual minimizer. Then, the optimal dual value is reached upon  $C_1^0\cap P_V$. Therefore,
$$u_{\simple}(0) =   \min_{y>0, Q\in C_1^0\cap P_V}  E\left[V\left( y \frac{dQ}{dP}\right)\right].  $$
The Lemmata \ref{simplesuper} and \ref{polar} rely on Assumption \ref{maximalS} to  give $ \measures = C_1^0\cap P_V$,  whence
the conclusion \eqref{strongdual2} follows.\qquad
\end{proof}

The above Theorem shows that the additional Assumption \ref{strongdual} amounts to requiring $\hat{z}_s=0$ for some dual optimizer $\hat{z}$ in \eqref{generaldual2}.
The next Corollary provides a simple sufficient condition which ensures that \emph{any} dual optimizer is regular.
\begin{corollary}\label{corstrongdual}
Let  $U$ be finite on the whole $\mathbb{R}$ and let $ S \in \SsigMU$. Under Assumption \ref{Inada}, for any $x\in (\underline{x},\xbar)$ the  simpler dual relation \eqref{strongdual2} holds. In other words,   Assumption \ref{strongdual} is automatically satisfied if Assumption \ref{Inada} holds and $S\in \SsigMU$.
\end{corollary}
\begin{proof}
Note first that the condition  $S \in \SsigMU$  may coincide with the generally weaker Assumption \ref{maximalS}.  This happens when $L^{\Uhat} = M^{\Uhat}$, that is when $U$ has left tail which goes to $-\infty$ at a ``moderate speed''. In such case, the dual space $(L^{\Uhat})^* $ is free of singular parts---exactly as in the dual system $(L^p,L^q)$ when $1\leq p<+\infty$---and Theorem \ref{generaldual} immediately yields the strong dual relation \eqref{strongdual2}.\\
  \indent So, suppose $S \in \SsigMU $ but $ M^{\Uhat} \subsetneq  L^{\Uhat} $.  The  most  intuitive  way to show \eqref{strongdual2}  is to note that terminal values $H\cdot S_T, H \in \simple$, are in $M^{\Uhat}$,   to set $\check{C}:=\{ k \in M^{\Uhat} \mid k \leq H\cdot S_T \text{ for  some } H \in \simple \}$ and to  work with  the dual system $( M^{\Uhat}, (M^{\Uhat})^* )$ instead of  the full $ ( L^{\Uhat}, (L^{\Uhat})^* )$. The advantage is that the elements of  $ (M^{\Uhat})^* $ are regular.  Then, an application of the  duality arguments   of Theorem \ref{generaldual} with $C$ replaced by $\check{C}$ gives
  $$ u_{\simple}(x) = \min_{y>0, Q\in (\check{C})_1^0\cap P_V} \left\{ xy + E\left[V\left( y \frac{dQ}{dP}\right)\right]\right\}. $$
  Now,  $(\check{C})^0_1 $ consists of probabilities and as in the final part of the Theorem   $(\check{C})^0_1\cap P_V =\measures $, whence \eqref{strongdual2}. \\
    \indent For the interested reader we provide an alternative proof which is less intuitive as it requires an analysis of the behavior of singular elements of $ (L^{\Uhat})^* $, but this proof makes direct use of the general dual formula \eqref{generaldual2}.   When
    $S\in \SsigMU$ the set  $C_1^0$ has a  special structure:
  $$   C_1^0 \ni  z= z_r+z_s \Leftrightarrow z_r \in C_1^0. $$
   This can be seen through the following steps: (i) $C^0$ coincides in fact with $ \{ z \in (L^{\Uhat})_+^*  \mid \langle z, H\cdot S_T \rangle = 0,\ \forall H \in \simple \}$, where the equality holds as $\simple$ is a vector space, and here $ H\cdot S_T \in M^{\Uhat}$;  (ii)  when $ U$ is finite on $ \mathbb{R}$, $\Uhat$ is also finite everywhere and with such Young functions   singular elements in the dual space   are null over  the Orlicz heart: if $z=z_s$ then $z$ is null over $M^{\Uhat}$; (iii) the Orlicz heart contains $L^{\infty}$; 4) consequently  $z \in C_1^0 $ iff  $z$ is a positive functional and
   \begin{eqnarray*}
      \langle z, I_{\Omega} \rangle &=& E\left[ \frac{dz_r}{dP} \right] + \langle z, I_{\Omega} \rangle =  E\left[ \frac{dz_r}{dP} \right] =1,  \\
      \langle z, H\cdot S_T \rangle &= &E\left[ \frac{dz_r}{dP}\, H\cdot S_T\right] + \langle z_s, H\cdot S_T\rangle =  E\left[ \frac{dz_r}{dP}\, H\cdot S_T\right]
      =  0  \text{ for all  } H\in \simple,
   \end{eqnarray*}
that is, iff $z_r \in C_1^0$. Now, a simple inspection of the dual problem in  \eqref{generaldual2} shows that,  for any $z\in C^0_1$,  setting $z_s$ to zero makes the dual function to be minimized smaller.   Hence, any minimizer is regular, i.e. we have shown \eqref{strongdual2}.\qquad
\end{proof}

\subsection{Characterization of the optimal solution: $\xbar=+\infty$, $\optQ$ not equivalent to $P$}\label{literature2}
When  $U$ is strictly monotone (a typical example is the exponential utility) but $\optQ$ is not equivalent to $P$ one can express the optimal terminal wealth $\optf$ using integrands in $L(S,\optQ)$ but no longer using the more natural strategies in $L(S,P)$. An \emph{approximation result} for $\optf$ via integrands in $L(S,P)$ was first shown by Acciaio
\cite{A05}, under the following technical conditions:
\begin{enumerate}
    \item[(i)]  $U$ is  differentiable, monotone, strictly concave
    and it satisfies $\mathrm{RAE}(U)$ (\ref{rae+}, \ref{rae-});
    \item[(ii)]  $S$ is locally bounded;
    \item[(iii)] the stopping times of the filtration are predictable.
\end{enumerate}
Acciaio builds a sequence of integrals $
\widetilde{H}_n\cdot S_T $, whose expected utility tends to the optimum, and
which satisfies $ (x + \widetilde{H}_n\cdot S_T) \rightarrow \optf $ $P$-a.s.

Our setup allows us to remove the
technical conditions above while proving $P$-a.s. convergence of terminal wealths in item (b.i) of Theorem \ref{main theorem}
and a stronger $L^1(P)$ convergence of utilities in item (b.ii),  which implies convergence of expected utility.


\section{Auxiliary results}\label{auxiliary}

\begin{lemma}\label{CC}
Let $\Psi:\mathbb{R}\rightarrow (-\infty, +\infty]$ be a convex, lower semicontinuous  function. For a given sequence $(x_n)_n$,
if  $d_n\in \mathbb{R}_+$, $\sum_{n\geq 1} d_n =1$ and $\sum_{n\geq 1} d_n x_n $ converges,
then
$$\Psi(\sum_{n\geq 1}  d_n x_n) \leq \liminf_N \sum_{n= 1}^N d_n\Psi(x_n).$$
When $\Psi$ is bounded from below, the above inequality simplifies to 
$$ \Psi(\sum_{n\geq 1}  d_n x_n) \leq  \sum_{n\geq  1} d_n\Psi(x_n).$$  
\end{lemma}
\begin{proof}From convexity of  $\Psi$,
 $$ \Psi(\sum_{n= 1}^N d_n x_n) \leq  (1- \sum_{n= 1}^N  d_n)\Psi(0) +  \sum_{n= 1}^N d_n \Psi(x_n). $$
When $N \uparrow +\infty$, $\sum_{n= 1}^N d_n x_n \rightarrow \sum_{n\geq 1} d_n
x_n$ so that lower semicontinuity of $\Psi$ implies $\Psi(
\sum_{n\geq 1}  d_n x_n)\leq \liminf_{N\rightarrow +\infty} \Psi(\sum_{n= 1}^N
d_n x_n)  $. The above displayed chain  shows that such $\liminf $ is dominated by
 $\liminf_N \sum_{n= 1}^N
d_n\Psi(x_n)$. Finally, when $\Psi$ is bounded from below, the latter series admits a limit (finite or $+\infty$).\qquad
\end{proof}

\begin{lemma}\label{ui}
Let $Q\ll P$. If $(Z^n \frac{dQ}{dP})_n$ is $P$-uniformly integrable, then  $(Z^n)_n$ is
$Q$-uniformly integrable, and vice versa.
\end{lemma}

{\em Proof}. This intuitive Lemma is a  consequence of the Dunford-Pettis criterion: 
A subset $K \subset L^1$ is uniformly integrable if and only if it is relatively compact for the weak topology. 
However,  here is  an elementary proof.  Recall $(X_\alpha)_\alpha$ is uniformly integrable when
$$ \lim_{r\to +\infty}\sup_\alpha E[|X_\alpha|\,I_{\{X_\alpha\geq r\}}]=0.$$
There is a well-known equivalent characterization of uniform integrability for   random variables: 
$(X_\alpha)_\alpha$ is uniformly integrable if and only if \emph{i)} the family is uniformly bounded 
in $L^1(P)$ and \emph{ii) }for every $\varepsilon >0$ there exists $\delta>0$ such that whenever $P(A)<\delta$,  
$ \sup_{\alpha} E [ I_A |X_{\alpha}|] <\varepsilon   $ (see e.g. the book \cite[Chapter 2.6]{sh}). 
So,   suppose $(Z^n \frac{dQ}{dP})_n$ is $P$-uniformly integrable. Then, for every $r>0$
\begin{align*}
 E_Q \left [I_{\{|Z^n|>r\}}
|Z^n| \right ]  &\leq
 E_Q\left [I_{\{|Z^n|>r , \frac{dQ}{dP}>\frac{1}{\sqrt{r}}\}} |Z^n|\right ]  +  E_Q\left[I_{\{ \frac{dQ}{dP}\leq \frac{1}{\sqrt{r}}
 \}} |Z^n| \right]  \\
 &\leq  E\left [I_{\{|Z^n| \frac{dQ}{dP}> \sqrt{r}\}} |Z^n|\frac{dQ}{dP} \right ]  + E\left [ I_{\{ 0< \frac{dQ}{dP}\leq \frac{1}{\sqrt{r}}
 \}} |Z^n|\frac{dQ }{dP} \right ],
\end{align*}
whence
\begin{align*}
\lim_{r\rightarrow+\infty }&  \sup_{n} E_Q \left [ I_{\{|Z^n|>r\}}|Z^n|\right]\\  
&\leq \lim_{r\rightarrow +\infty }  \sup_{n} \left (
E\left [I_{\{|Z^n| \frac{dQ}{dP}> \sqrt{r}\}} |Z^n|\frac{dQ}{dP} \right]  +
E \left[I_{\{ 0 < \frac{dQ}{dP}\leq \frac{1}{\sqrt{r}}
 \}} |Z^n|\frac{dQ }{dP}\right ] \right )=0,
\end{align*} 
where the last equality follows from $P$-uniform
integrability of $ (Z^n \frac{dQ}{dP})_n$ and from the fact  that   $\{
0<  \frac{dQ}{dP}\leq \frac{1}{\sqrt{r}}\} $ has $P$-probability
which tends to $0$ when $r$ goes to $+\infty$.

The converse implication follows directly  from $Q\ll P$:
\begin{align*}
\lim_{P(A)\to 0} \sup_n E\left[I_A|Z^n|\frac{dQ}{dP}\right] &= \lim_{P(A)\to 0} \sup_n E_Q[I_A|Z^n|]\\
&\leq \lim_{Q(A)\to 0} \sup_n E_Q[I_A|Z^n|]=0.\qquad\endproof
\end{align*}
{\em Proof of Lemma \ref{dualityQ}}. $u_Q(x)<+\infty$ follows from Fenchel inequality and from finite generalized relative entropy of $Q$: if $X$ satisfies $E_Q[X]\leq x$,
$E[U(X)] \leq xy_Q + v_Q(y_Q)$, with $y_Q$ from Definition \ref{P-V}. The dual formula to be proved  is actually a straightforward consequence of the Fenchel duality formula and of the  results obtained by Rockafellar in the 1970-ies on conjugates of functionals in integral form (here, expected utility). However, we give a different proof based on Orlicz duality, since it is  useful for Theorem \ref{generaldual} where the Orlicz setup is necessary.

         The utility maximization problem $\sup_{E_Q[X]\leq x} E[U(X)]$ can be rewritten
         over the utility-induced Orlicz space $\Orlicz(P)$ defined in
        \eqref{Orlicz}. This can be done because: i) the supremum will be
        reached over those $X$ such that $E[U(X)]$ is finite, so that
        $-X^- \in \Orlicz(P)$; ii) if  $E[U(-X^-)]>-\infty $
        then the truncated sequence $X_n = X\wedge n$ is also in the
        Orlicz space and by Fatou Lemma in the limit it delivers the same
        expected utility from $X$;   iii) $\Orlicz(P) \subseteq L^1(Q)$, which follows from   $Q\in
        P_V$, from  \eqref{uhat-u-orlicz} and
        Fenchel inequality (this also implies $Q$ is in the topological dual of
        $L^{\Uhat}$).
          Therefore,
        $ u_Q(x) = \sup_{X \in \Orlicz, E_Q[X]\leq x}  E[U(X)]
        $. On $L^{\Uhat}$, the concave functional $I_U(X):=
        E[U(X)]$ is proper:
          $$ X\in L^{\Uhat} \Rightarrow X\in L^1(P) \text{ so that } E[U(X)]\stackrel{\mathrm{Jensen}}{\leq} U(E[X])< +\infty.  $$
          Moreover, $I_U$    has a  continuity
        point which belongs to the maximization domain $D=\{ X\in L^{\Uhat} \mid E_Q[X]  \leq x
        \}$. This is more subtle to check, but it can be proved
        that  the set
        $$\mathcal{B}:=\{ X \in L^{\Uhat} \mid E[ U(-{(1+\epsilon)}X^- )]>- \infty  \text{ for some } \epsilon >0\},$$
        coincides with the interior of the proper domain of
        $I_U$ (see \cite[Lemma 4.1]{bfg08} modulo a sign change),
         where $I_U$ is automatically continuous by the Extended Namioka Theorem (see e.g. \cite{bf09}). Then, as
           $x>  \underline{x}  $, the constant $x$  is  in $\mathcal{B}\cap D$. \\
        The   dual formula \eqref{uQdual1}    is thus  a consequence of Fenchel
        Duality Theorem \cite[Chapter 1]{BRE},  of
        the fact that the polar set  of the  constraint $C:=\{X\mid
        E_Q[X]\leq x\}\supseteq -L^{\infty}_+$, i.e. the set $\{ \mu \in
        (\Orlicz)^* \mid \mu(X) \leq x \, \forall X \in C \}$, by the
        Bipolar Theorem is the positive ray $\{yQ \mid y\geq 0\}$,
        and of the expression of the convex conjugate $(I_U)^*$ of $I_U$
        over  the  variables $y\frac{dQ}{dP}$: $(I_U)^*( y\frac{dQ}{dP}) = E[V
        (y\frac{dQ}{dP})] =v_Q(y)$.\qquad\endproof
\begin{proposition}\label{LQ-boundedness}
Suppose $(k_n)_n$
is a sequence of random variables such that $(E[U(k_n)])_n$ is
bounded from below and assume $(E_{\widetilde Q}[k_n])_n$ is bounded from above
for some $\widetilde Q\in P_V$ satisfying  the Inada condition \eqref{Inada2}.
Then the following statements hold:
\begin{enumerate}
\item[{\em(i)}] $U(k_n)$ is $L^1(P)$-bounded; 
\item[{\em(ii)}] $k_n$ is $L^1(Q)$-bounded for \emph{any} $Q\in P_V$ for which $E_{Q}[k_n]$
is bounded from above. The indirect utility $u_Q$ need not satisfy the Inada condition
\eqref{Inada2}.
\end{enumerate}
\end{proposition}
{\em Proof}. In this proof $c$ refers to a constant, not necessarily the same on each line.
\begin{enumerate}
\item[(i)] By hypothesis there is $0<y_1<y_2$ such that $v_{\widetilde
    Q}(y_i)<+\infty $ for $i=1,2$. The Fenchel inequality
    implies
    \begin{align}
    E[U(-k_n^-)]&\leq v_{\widetilde Q}(y_2) - y_2 E_{\widetilde Q}[k_n^-],\label{aux3}\\
    E[U(k_n^+)]&\leq v_{\widetilde Q}(y_1) + y_1 E_{\widetilde Q}[k_n^+],\label{aux4}
    \end{align}
    which yields
    \begin{equation}\label{aux5}
    E[U(k_n)] \leq c + y_1 E_{\widetilde Q}[k_n] - (y_2-y_1)E_{\widetilde Q}[k_n^-].
    \end{equation}
    By assumption, $(E_{\widetilde Q}[k_n])_n $ is bounded from above and $
    (E[U(k_n)])_n$ is bounded from below, whereby one concludes from
    \eqref{aux5} and from $y_2-y_1>0$ that $(E_{\widetilde
    Q}[k_n^-])_n$ is bounded and consequently  $(E_{\widetilde
    Q}[k_n^+])_n$ is also bounded. Finally, by \eqref{aux4} the
    sequence $(E[U(k_n^+)])_n$ is bounded. Since $U(k_n^+)\geq 0,
    U(-k_n^-)\leq 0$ and $(E[U(k_n)])_n$ is bounded from below the
    $L^1(P)$-boundedness of $U(k_n)$ follows.
\item[(ii)]  The inequality \eqref{aux3} applies for any $Q\in P_V$, i.e. there
    is $y_Q>0$ such that
    \begin{equation}\label{aux6}
    E[U(-k_n^-)]\leq c - y_{Q} E_{Q}[k_n^-].
    \end{equation}
    By (i) the sequence $(E[U(-k_n^-)])_n$ is bounded from below whereby
    $(E_{Q}[k_n^-])_n$ must be bounded. As in (i), this and boundedness
    from above of the expectations $(E_ Q[k_n])_n$ ensure
    $(E_Q[k_n^+])_n$ is also bounded.\qquad\endproof
\end{enumerate}

\begin{lemma} \label{polar}
Let $Q\in P_V$ verify $E_Q[X_T]= 0$ for all $X = H\cdot
S, H\in \simple$. Then $Q\in \measures$.
\end{lemma}
\begin{proof}
We just need to show $Q\in \mathcal{M}$.
Consider $S\in\SsigU\setminus
\scr{S}^{\Uhat}_{\mathrm{loc}}$, fix any $\I$-localizing
$\varphi$  from  Proposition \ref{sigmaI} and let $S' =
\varphi\cdot S $. For any  $A\in \mathcal{F}_s, s \in [0,T[, t>s $
let $H = I_A  I_{] s , t]}\varphi$, which is in $ \simple$. Since
$H\cdot S = (I_A  I_{] s , t]} )\cdot S'$ and $E_Q[ H\cdot S_T] =
E_Q[ I_A (S'_t-S'_s) ]=0$, for all   $A\in \mathcal{F}_s, s<t$,
$S'$ is a  then $Q$-martingale, and hence
$Q\in \mathcal{M}$. For $S\in \scr{S}^{\Uhat}_{\mathrm{loc}}$ we proceed as above, replacing $\varphi$ with $I_{[0,\tau_n]}$.\qquad
\end{proof}
\section*{Acknowledgements} We would like to thank Jan
Kallsen, Fabio Maccheroni, Mark Owen, Antonis Papapantoleon,  Maurizio Pratelli, Mihai S\^{\i}rbu and two anonymous referees for
helpful comments.

\end{document}